\documentclass[
a4paper, 		
12pt, 			
]{article}

\usepackage[utf8]{inputenc}

\usepackage{lmodern} 
\usepackage{chngcntr} 
\usepackage{abstract} 
\usepackage{microtype} 
\usepackage{ulem} 
\usepackage[utf8]{inputenc} 
\usepackage[T1]{fontenc} 
\usepackage[a4paper, left=2cm, right=2cm, top=3cm, bottom=3cm]{geometry}
\usepackage{graphicx} 
\usepackage{tikz}
\usepackage{pgfplots}
\pgfplotsset{compat=1.8}
\usepackage{float}
\usepackage{caption}
\usepackage[font={small,it}]{caption}

\usepackage{amsmath}
\usepackage{amsthm}
\usepackage{amsfonts}
\usepackage{nicefrac} 
\usepackage{gensymb} 

\usepackage{setspace}
\usepackage{color}

\newtheorem*{remark}{Remark}
\newtheorem{theorem}{Theorem}
\usepackage{titlesec}

\usepackage{algorithm}
\usepackage{algpseudocode}

\title{Probabilistic Constrained Bayesian Inversion for Transpiration Cooling}
\date{March 2022}
\usepackage{blindtext}

\author{Ella Steins \thanks{IRTG Modern Inverse Problems, RWTH Aachen University, Schinkelstr. 2, Aachen, steins@aices.rwth-aachen.de}, Tan Bui-Thanh \thanks{The Oden Institute for Computational Engineering \& Sciences, and the Department of Aerospace Engineering \& Engineering Mechanics, tanbui@oden.utexas.edu}, Michael Herty \thanks{Institute for Geometry and Practical Mathematics (IGPM), RWTH Aachen University, Templergraben 55, Aachen, herty@igpm.rwth-aachen.de, mueller@igpm.rwth-aachen.de}, Siegfried Müller \footnotemark[3]}
\begin{document}
	\maketitle
	\begin{abstract}              
		To enable safe operations in applications such as rocket combustion chambers, the materials require cooling to avoid material damage. Here, transpiration cooling is a promising cooling technique. Numerous studies investigate possibilities to simulate and evaluate the complex cooling mechanism. One naturally arising question is the amount of coolant required to ensure a safe operation. To study this, we introduce an approach that determines the posterior probability distribution of the Reynolds number using an inverse problem and constraining the maximum temperature of the system under parameter uncertainties. Mathematically, this chance inequality constraint is dealt with by a generalized Polynomial Chaos expansion of the system. The posterior distribution will be evaluated by different Markov Chain Monte Carlo based methods. A novel method for the constrained case is proposed and tested among others on two-dimensional transpiration cooling models.
	\end{abstract}
	\setcounter{tocdepth}{1}
	\section{Introduction}
	Efficient cooling concepts are needed in all applications, where materials have to be protected against thermal damage as a result of high thermal loads. Passive cooling for instance uses coating of the material as a heat shield  \cite{arai2013, esser2016innovative, Langener2011}.
	We are interested in models for active cooling techniques, like film and transpiration cooling (TC) where a coolant is used that transfers heat away from the materials in need of protection. These type of methods generally have a great cooling potential \cite{esser2016innovative}. Here, transpiration cooling is designed as the injection of a coolant into a hot gas flow through a porous medium. Cooling is achieved by both convection energy within the wall while the coolant passes through the porous medium as well as a coolant film that insulates the wall from the hot gas flow \cite{Langener2011}. \\
	Some applications nowadays considered for transpiration cooling include scramjet combustion chambers \cite{dlr129995}, rocket thrust chambers \cite{Dahmen2014}, turbine blades \cite{huang2017experimental} and others \cite{arai2013, boehrk2014transpiration}. Even so it has been invented in the 1950's \cite{Eckert1954ComparisonOE}, due to both advances made in manufacturing and the need for greater cooling efficiency, recently transpiration cooling has come back into the focus of research and has been intensively studied experimentally, analytically and numerically \cite{Dahmen2014, selzer2009production, wang2004experimental, Boehrk2010, DING2019422, leontiev2019effect, christopher2020dns}. \\
	
	The porous materials used for transpiration cooling have to meet certain requirements, like high porosity, light-weight, precise geometry and additionally, for turbine blades \cite{huang2017experimental},  great material strength. Commonly used types of porous materials are sintered metal or ceramic porous media \cite{boehrk2014transpiration} and ceramic matrix composites \cite{reimer2011transpiration}, leaving the pores to be a result of randomness. This in turn leads to negative cooling effects like badly connected pores and demands also suitable simulation techniques to quantify its effect \cite{Selzer2014}. Additionally, effects of non-uniform permeability have been investigated \cite{wu2018optimization}. Recently, the improved precision of the manufacturing process by additive manufacturing methods has been studied \cite{huang2018transpiration, min2019experimental}. \\
	
	For simulation and optimization of transpiration cooling, temperature models \cite{Boehrk2010, liu2013effects} and coupled simulations \cite{DING2019422, Dahmen2014} are used. The point of departure of this work is the publication by Dahmen et al. \cite{Dahmen2014}. Therein, numerical simulation of transpiration cooling by coupling a porous media flow with a hot gas flow of a rocket combustion chamber is proposed. Numerical experiments in both 2D and 3D are reported. In the following sections we also focus on the interaction of the coolant with the temperature to allow for the treatment of e.g. temperature constraints. Alternative detailed injection models are based on turbulence modeling, heuristic considerations, and DNS simulations. \cite{wilcox2006turbulence, Koenig, christopher2020dns, cerminara}. \\
	
	We propose to consider a model with uncertain parameters to acknowledge for the uncertainties present in the production of the porous material and the simulation of the flow as well as the so far missing detailed physical model for the interaction of flow and porous media and the presence of inequality requirements for the temperature. \\ Parametric uncertainties can be handled by non-intrusive methods \cite{mackay1998introduction, xiu2005high} that repeatedly evaluate the deterministic model such that statistical properties of the system's answer can be derived. \\
	
	Generalized polynomial chaos (gPC) is an intrusive method that represents the full output probability through a spectral expansion with orthogonal basis polynomials \cite{wiener1938homogeneous, cameron1947orthogonal}. In \cite{ghanem1998probabilistic} stochastic finite element systems for probabilistic transport in porous media are considered.  Xiu and Karniadakis extended the polynomial chaos expansion to various orthogonal basis polynomials \cite{xiu2002wiener}. Further, gPC can be used to perform global sensitivity analysis \cite{sudret2008global}. \\
	
	As the evaluation of the temperature constraint in TC requires the full probability distribution, we propose here a gPC expansion. At the same time, optimizing the cooling effect by choosing a suitable Reynolds number while using the pressure as a simulation output, is formulated as an inverse problem. The latter is solved using a framework based on Bayesian Inversion. \\ 
	Even so research devoted to Bayesian Inverse Problems is available, we only review literature on the constrained case. Wu et al. extend Bayesian inversion by incorporating an additional likelihood to the Bayesian Inversion based on fitness of the solution to the constraint \cite{Wu2019}. \\
	Another approach is presented in the context of Bayesian optimization in the machine learning community. Here, in order to globally optimize black-box derivative-free methods, a statistical Gaussian surrogate for the objective function is built upon Gaussian process regression \cite{Frazier2018}. An acquisition function is derived from the surrogate and extended to the constrained optimization problems in \cite{Gardner2014, Gelbert2015}. \\
	
	The paper is organized as follows: In Section 2, the two models for transpiration cooling are presented and the extension to the stochastic system is given. Section 3 then introduces the solution method, a novel constrained random walk Markov Chain Monte Carlo method, and two alternative solution methods for comparison. Following this, the numerical results are discussed in Section 4. 
	
	\section{Modeling transpiration cooling at an interface}
	
	\textbf{Notation} In the following sections, all random variables (RV) $\boldsymbol{X}$ are given in bold type and realizations are $X$. We assume each RV $\boldsymbol{X}$ has a probability density denoted by $\Pi(X)$. Data is denoted by the suffix $^{data}$. \\ \\
	
	The dimensionless \textit{forward problem} consists of two coupled ordinary differential equations (ODEs), one for the fluid temperature $T_f=T_f(x; \boldsymbol{Re})$, one for the solid temperature $T_s(x; \boldsymbol{Re})$ of a spatially 1-D strip porous media of normalized length $x \in [0,1]$. Here, $\boldsymbol{Re}$ is the uncertain Reynolds number\footnote{The Reynolds number $Re := \frac{\dot{m}L}{A \mu}$ for $\dot{m}, L, A, \mu$ corresponds to the mass flow in the dimensionless system.}. The third ODE describes the evolution of the density $\rho_f$ of the coolant, where $\rho_f'(x, \boldsymbol{Re})=\frac{d\rho_f(x, \boldsymbol{Re})}{dx}$. The velocity of the coolant is denoted by $v(x)$, $x \in [0,1]$.
	\begin{subequations}
		\begin{align}
			\begin{split} 
				T_s'(x; \boldsymbol{Re})&=\frac{\kappa_f}{(1-\varphi)\cdot \kappa_s}\boldsymbol{Re}\cdot Pr_f\cdot (T_f(x; \boldsymbol{Re})-T_{HG})+\frac{q}{(1-\phi)\cdot \kappa_s}, 
				\label{2a}
			\end{split}
			\\
			\begin{split}
				T_f'(x; \boldsymbol{Re})&=\frac{Nu_{v,f}}{Pr_f \cdot \boldsymbol{Re}}\cdot (T_s(x; \boldsymbol{Re})-T_f(x; \boldsymbol{Re})), 
				\label{2b}
			\end{split}
			\\
			\begin{split}
				\rho_f'(x; \boldsymbol{Re}) &= N(x; \boldsymbol{Re}) \cdot \rho_f(x; \boldsymbol{Re})
				\label{2c}
			\end{split}
			\\
			\begin{split}
				\intertext{with}
				N(x; \boldsymbol{Re}) &= \dfrac{\frac{Nu_{v,f}}{\boldsymbol{Re} \cdot Pr_{f}} \rho_f^2(x; \boldsymbol{Re}) \left(T_s(x; \boldsymbol{Re})-T_f(x; \boldsymbol{Re})\right)+\left(\frac{L^2}{\boldsymbol{Re} K_D}+\frac{L}{K_F} \right)}{\phi^{-2}-\rho_f^2(x; \boldsymbol{Re})T_f(x; \boldsymbol{Re})}, \nonumber 
			\end{split}
			\\
			\begin{split}
				v(x; \boldsymbol{Re})&=\frac{1}{\rho_f(x; \boldsymbol{Re})}.
				\label{2d}
			\end{split}
			\\
			\intertext{The pressure at the interface is denoted by $p$ and}
			\begin{split}
				\boldsymbol{p} :&=p_f(x=1; \boldsymbol{Re})=T_f(x=1; \boldsymbol{Re}) \cdot \rho_f(x=1; \boldsymbol{Re})
				\label{2e}
			\end{split}
			\\
			\intertext{and the inititial conditions are given by} 
			\begin{split}
				T_f(0; \boldsymbol{Re}) &=T_c \, , \, T_s(0; \boldsymbol{Re})=T_b \, , \, \rho_f(0; \boldsymbol{Re})=\frac{p_R}{T_c}.
				\label{2f}
			\end{split}
		\end{align}
	\end{subequations}
	For this system (\ref{2a})-(\ref{2f}), the other model parameters are given in Table \ref{tab_set} in the Appendix. \\
	The \textit{deterministic forward problem} is summarized as 
	\begin{align}
		\boldsymbol{p}:= F(\boldsymbol{Re}),
		\label{1}
	\end{align}
	where the pressure $\boldsymbol{p}:=p_f(x=1; Re)$ in (\ref{2e}) is observable at the position $x=1$. Observations of the pressure are denoted by $p^{data}$ and are drawn from a numerical integration of equations (\ref{2a})-(\ref{2f}) of the coolant through a porous material with superposed white noise. The physical relations are described by the operator $F$ which is comprised of a differential-algebraic system of equations (DAE). \\
	We will first consider the spatially homogeneous situation, where we refer to (\ref{2a})-(\ref{2f}) as \textit{model 1}. The constraint $T_f(x=1; Re) \leq T_{max}$ will then be imposed on the temperature $T_f$. \\ 
	In a second step, we extend the problem to the non-homogeneous case, here referred to as \textit{model 2}. We consider a porous media with elongation $z \in (d_1, d_2)$ divided in two sections, where each section consists of finitely many vertical 1D pores as indicated in Figure \ref{sketch}. Each vertical strip is parametrized by $x \in (0,1)$. Both sections have their own deterministic porosity $\phi_j, j=1,2$; but the \textit{same} uncertain Reynolds number $\boldsymbol{Re}$. This geometry is common, see \cite{Peichl2021}. At the boundary of the porous media we assume solid material with a possibly high initial temperature $T_0 > T_{max}$. We are interested in the evolution of the temperature $T_h(z, t)$ at the interface of the porous material with the surrounding conditions over time $t$. In our model this temperature is obtained from the porous media equations at $x=1$. \\
	The temperatures $T_h(z_i, t=0) = T_f^i(x=1; \boldsymbol{Re})$ at different locations $z_i \in (0,1) \, , i=1,...,N$, are then given by \textit{model 1} (\ref{2a})-(\ref{2f}) with porosity $\phi_i$. Here, we denote by $T_f^i(x=1; \boldsymbol{Re})$ the solution to (\ref{2a})-(\ref{2f}) for fixed porosity $\phi_j$. Across the interface $z$ the evolution of the temperature is assumed to be given by the linear heat diffusion equation with known deterministic diffusion coefficient $\lambda > 0$: 
	\begin{subequations}
		\begin{align}
			\begin{split}
				0 &= \frac{\partial T_h(z,t)}{\partial t}-\lambda \frac{\partial^2 T_h(z,t)}{\partial z^2}, \, \, \, \, z \in [0,1], \, t \geq 0.
				\label{hd}
			\end{split}
			\intertext{The initial condition is given by the temperature obtained by the 1D strip, see (\ref{hd_rb}), where we assume the porous media pores are of size $\Delta z$:}
			\begin{split}
				T_h(z, t=0) &= \left \{  \begin{matrix}  \sum_{i=d_1}^{d_2} \chi_{[z_i-\Delta z, z_i+\Delta z]}(z) T_f(x=1; \boldsymbol{Re}) &  z \in(d_1, d_2) \subset [0,1] \\ T_0 &\text{ otherwise} \end{matrix} \right.
				\label{hd_rb}
			\end{split}
		\end{align}
	\end{subequations}
	and Neumann boundary conditions at $z=0$ and $z=1$, respectively. \\ The constraint will be imposed at a fixed terminal time $t=t_c$ and hence reads
	\begin{equation}
		T_h(z, t=t_c) \leq T_{max} \, \, \forall z \in [0,1].
		\label{hd_constraint}
	\end{equation}
	
	{\begin{figure}[h]
			\centering
			\includegraphics[width=1\linewidth]{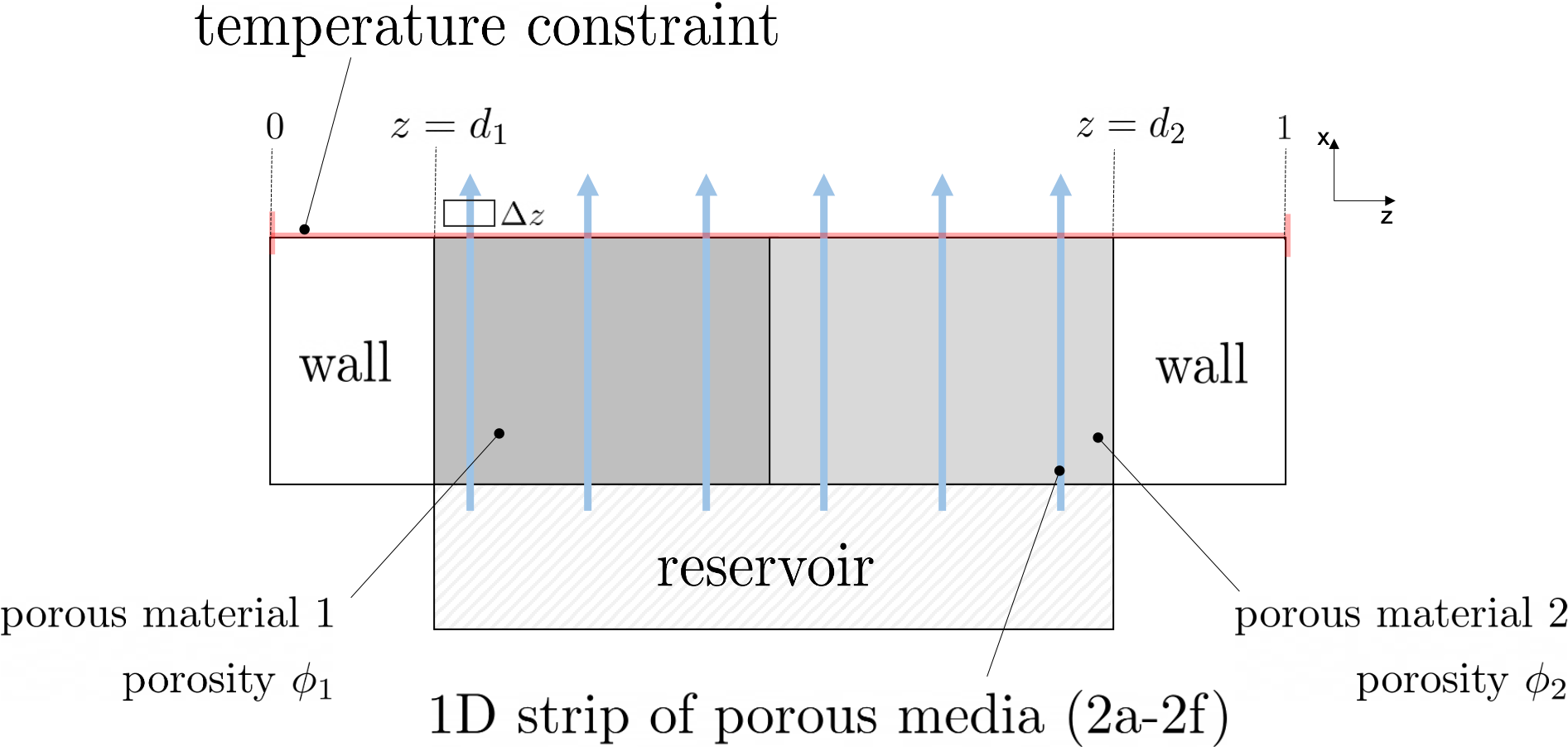}
			\caption{This figure shows the setup of \textit{model 2} and \textit{model 3}. Two porous materials with different porosities are considered, and sixty 1D strip simulations are performed to derive the spatially-dependent temperature at the interface. Across $z$, the evolution of the temperature is assumed to be given by the heat equation. The temperature constraint is imposed over the whole length of $z$ at some time $t=t_c > t_0$.}
			\label{sketch}
	\end{figure} }
	
	Incorporating model uncertainties into the design is crucial for substantively design and risk assessment The porosity $\pmb{\phi}$, respectively, is a result of a manufacturing process where the porous material breaks at random spots \cite{min2019experimental}. The heat flux $\boldsymbol{q}$ is random itself as a product of the randomness of the porous material. A sensitivity analysis of the model showed the highest sensitivity of the temperature towards uncertainties in $(\boldsymbol{Re}, \boldsymbol{\phi_i}, \boldsymbol{q})$. For the one-dimensional model, the heat flux $\boldsymbol{q}$ and the porosity $\pmb{\phi}$ will be treated as parametric uncertainties. For \textit{model 2} we can account for a spatial uncertainty by modeling two sections with separate determinsitic porosities but the same uncertain heat flux $\boldsymbol{q}$. Lastly, \textit{model 3} additionally incorporates spatially uncertain heat flux with RVs $\boldsymbol{q_1}, ..., \boldsymbol{q_{N}}$. \\ Even so the following discussion is not limited to this case, we assume normally distributed uncertainties, with $\boldsymbol{q} \sim \mathcal{N}(q_0, \sigma_q^2)$ and $\boldsymbol{\phi} \sim \mathcal{N}(\phi_0, \sigma_\phi^2)$. \\
	
	As the parametric uncertainties of the heat flux and the porosity enter the temperature, $T_f$ and $T_h$ also become RVs, i.e. $\boldsymbol{T}_f=\boldsymbol{T}_f(x,\boldsymbol{q},\boldsymbol{\phi}; \boldsymbol{Re})$. To formalize, we assume that $(\boldsymbol{q},\boldsymbol{\phi})$ is a RV defined on a probability space $(\Omega, P(\Omega); \mathbb{P})$ mapping to $\mathbb{R}^2$. \\ For each fixed $(x, Re) \in (0,1) \times \mathbb{R}^+$ we additionally assume that $\boldsymbol{T}_f(x,\cdot, \cdot; Re) \in L^2(\Omega; \mathbb{R})$. 
	We apply a gPC expansion to the system of ODEs governing the temperature evolution in the porous media. Denote by $\mu_1(q)dq$ and $\mu_2(\phi)d\phi$ the probability measures of the independent RVs $\boldsymbol{q}$ and $\boldsymbol{\phi}$, respectively. Then, we define for $d\nu(q, \phi) = \mu(q)\mu_2(\phi)dqd\phi$ and $y_i \in L_v^2(\Omega; \mathbb{R})$ the scalar product
	\begin{equation}
		< y_1 ,  y_2> = \int_{-\infty}^{+\infty} \int_{-\infty}^{+\infty} y_1\left(\tilde{q}, \tilde{\phi} \right) \cdot y_2\left(\tilde{q}, \tilde{\phi} \right) d\nu(\tilde{q}, \tilde{\phi}).
		\label{Hermite_orthotogonal1}
	\end{equation}
	Then, for any $y \in L_v^2(\Omega; \mathbb{R})$ we have $|| \sum_{k=0}^K \bar{y}_k \cdot \tilde{\Phi}_k(q,\phi) - y||_{L_v^2} \rightarrow 0, K\rightarrow\infty,$ where $\bar{y}_k$ is the Galerkin projection onto the space spanned by the polynomials $\tilde{\Phi}_0, ..., \tilde{\Phi}_N$. The previous expansion is the gPC expansion and we apply it to $\boldsymbol{T}_f(x,\boldsymbol{q},\boldsymbol{\phi}; Re)$. Since $(\boldsymbol{q}, \boldsymbol{\phi})$ are independent RV we have after possible normalization that $\tilde{\Phi}_k(q,\phi)=\Phi_k(q) \Phi_k(\phi)$. \\
	The coefficients $\hat{T}_{ij}$ are determined by solving
	\begin{subequations}
		\begin{equation}
			\frac{d}{dx} \hat{T^f}_{i,j} = \frac{1}{<\Phi_i^2, \Phi_j^2>} \int_{-\infty}^{\infty}  \int_{-\infty}^{\infty} \frac{Nu}{Re Pr} \left( T^K_s - T^K_f \right) d\phi dq, 
		\end{equation}
		\begin{equation}
			\frac{d}{dx} \hat{T^s}_{i,j}= \frac{1}{<\Phi_i^2, \Phi_j^2>} \int_{-\infty}^{\infty}  \int_{-\infty}^{\infty} \frac{\kappa_f}{(1-\phi)\kappa_s} Re Pr \left( T^K_f - T_{hg} \right) + \frac{q}{1-\phi \, \kappa_s} d\phi dq
		\end{equation}
		where $T^K_s=\sum \limits_{i,j=0}^{K} \hat{T^s}_{ij}(x; Re) \Phi_i(q) \Phi_j(\phi) $, $T_f^K=\sum \limits_{i,j=0}^{K} \hat{T^f}_{ij}(x; Re) \Phi_i(q) \Phi_j(\phi)$,
	\end{subequations}
	\[
	\hat{T^f}^{0}_{i,j}= \left \{  \begin{matrix} T_c &  i=j=1  \\ 0 &\text{ otherwise} \end{matrix} \right. \text{ and similarly }
	\hat{T^s}^{0}_{i,j}= \left \{  \begin{matrix} T_b &  i=j=1  \\ 0 &\text{ otherwise} \end{matrix} \right..
	\]
	Numerically, an explicit Euler scheme is used in Section 4 as well as stochastic collocation to approximate the integrals. The gPC expansion of the temperature $\boldsymbol{T}_h(z,t)$ can be obtained using also a gPC expansion of eq. (\ref{hd}).

	\section{Solution Method}
	In this section, the framework of probabilistic constrained Bayesian Inversion is introduced. A solution method based on a combination of gPC and constrained Markov Chain Monte Carlo is proposed. The method is explained in the general setting for notational brevity and clarity. In Section 4 it will be applied to the transpiration cooling problem. \\ \\
	In the following section the finite dimensional optimization parameter is $\theta \in \mathbb{R}^{N_\theta}$. Let $(\Omega, F(\Omega), \mathbb{P})$ be a probability space and $\boldsymbol{\bar{\xi}}:\Omega \rightarrow \mathbb{R}^d$ a RV with Lebesque probability density $\Pi$. For a deterministic function $\bar{f}: \mathbb{R}^d \rightarrow \mathbb{R}^{N_f}$ we have in case of random input $\bar{f}(\bar{\xi}(\omega))$ for $\omega \in \Omega$. A change of variables $\bar{\xi}(w) \rightarrow \xi$ gives the considered representation $\bar{f}(\xi): \mathbb{R}^d \rightarrow \mathbb{R}^N_f$.
	
	\subsection{Problem formulation} 
	In the considered case we denote by $f_1: \mathbb{R}^{N_\theta} \times \mathbb{R}^d \rightarrow \mathbb{R}$ a forward model depending on an optimization parameter $\theta$ as well as the RV $\boldsymbol{\xi} = \boldsymbol{\bar{\xi}} \in \mathbb{R}$. Similarly, we assume that $f_2: \mathbb{R}^{N_\theta} \times \mathbb{R}^d \rightarrow \mathbb{R}$ models possible point wise inequalities. \\ \\
	As in a Bayesian framework, we assume given prior information as probability density of the parameter $\Pi_{prior}(\theta)$. Let $f_1$ be observable with RV $\boldsymbol{f_1^{data}}=f_1(\xi^*, \theta^*)+\boldsymbol{\eta}$ with variance $\sigma_l^2$, where $\boldsymbol{\eta}$ is white noise, $\theta^* \in \mathbb{R}^{N_\theta}$ is an unknown parameter, and $\xi^* \in \mathbb{R}^d$ an unknown realization. As in e.g. \cite{Wu2019}, the Likelihood function $\Pi_l: \mathbb{R}^{N_f} \rightarrow \mathbb{R}$ is defined by
	\begin{equation}
		\Pi_l(f_1^{data}|\theta) := \frac{1}{\sqrt{2 \pi} \sigma_l} \cdot e^{\frac{-1}{2N\sigma_l^2} \sum_{i=1}^N (f_{1,i}^{data}-f_1(\mathbb{E}_\xi, \theta))^2}.
		\label{likelihood}
	\end{equation}
	Here, $f_{1,i}^{data}$ are $N$ realizations of RV $\boldsymbol{f_1^{data}}$. The model output $f_1$ is computed with respect to the expectation $\mathbb{E}_\xi=\int \xi \Pi_\xi(\xi) d\xi$, as for the Bayesian Inversion only the parameter $\boldsymbol{\theta}$ is optimized. 
	The optimization of this parameter has the goal of finding realizations which will be used in the operation of the cooling mechanism. For example, the mass flow of the coolant can be adjusted. Since an optimality for a single parameter is not desired from an application point of view, we are interested in the optimal probability distribution $\Pi(\theta)$ of this parameter $\theta$. The associated RV is $\boldsymbol{\theta}$.\\
	In contrast to that we have parametric uncertainties, labeled $\boldsymbol{\xi}$. The porosity, for example, is a RV as the manufacturing process is inherently uncertain and thus not controllable.

	These parametric uncertainties come into play as for some fixed threshold $\beta$ the risk that $f_2(\boldsymbol{\xi}; \theta) \leq \beta$ is not satisfied in face of the parametric uncertainties $\boldsymbol{\xi}$ for the optimized realization $\theta$ should be lower than $\alpha$. \\

	Thus, the goal is to determine the constrained posterior $\Pi^c_{post}(\theta)$ such that the probability $\mathbb{P}_{\xi}(f_2(\boldsymbol{\xi}; \theta) \leq \beta) \geq \alpha$ a.s. w.r.t. $\boldsymbol{\theta}$. Here, $\alpha$ is our risk. We propose now to obtain such a (posterior) distribution $\Pi^c_{post}(\theta)$ by restricting the admissible set for $\theta \in \mathbb{R}^{N_\theta}$ such that the constraint is fulfilled: Define 
	\begin{equation}
		\mathcal{S} := \{ \theta \in  \mathbb{R}^{N_\theta}: \mathbb{P}_{\xi}(f_2(\boldsymbol{\xi}; \theta) \leq \beta) \geq \alpha \}
		\label{subset}
	\end{equation}
	as the subset for which the constraint is fulfilled. Then, the constrained posterior can be expressed as 

	\begin{equation}
		\Pi^c_{post}(\theta) = c \cdot \mathcal{X}_{\mathcal{S}}(\theta) \cdot \Pi_{prior}(\theta) \cdot \Pi_l(f_1^{data}|\theta)
		\label{g_optimization}
	\end{equation}
	where $\mathcal{X}_{\mathcal{S}}$ is the characteristic function on the set ${\mathcal{S}}$ and $c$ is a normalization constant. \\
	
	Numerically, the problem (\ref{g_optimization}) requires an efficient description of the set ${\mathcal{S}}$ of the constraints.
	We propose using a gPC expansion of $f_2=\sum_{i=0}^\infty \hat{f}_{2,i} \Phi_i(\xi)$ in $\xi$. This leads for any fixed realization $\theta \in \mathbb{R}^{N_\theta}$ to
	\begin{equation}
		\mathbb{P}_{\boldsymbol{\xi}}(f_2(\boldsymbol{\xi}; \theta) \leq \beta) = \int_{\mathbb{R}^d} \left \{ \begin{matrix} 1 &\text{ if }  \sum_{i=0}^\infty \hat{f_2}_i(\theta) \Phi_i(\xi) \leq \beta \\ 0 &\text{ otherwise} \end{matrix} \right \}  \Pi(\xi) d\xi.
		\label{g_constraint_int}
	\end{equation}
	
	Hence, the set $\mathcal{S}$ is given by $\mathcal{S} = \{ \theta \in \mathbb{R}^{N_\theta}: \int_{\mathbb{R}^d} \{\} \Pi(\xi) d\xi \geq \alpha \}$.
	
	\begin{remark} 
		In a Bayesian setting, an expansion of $f_2$ in $(\boldsymbol{\xi}, \boldsymbol{\theta})$ would return a global cheap-to-evaluate meta model. However, this is not accurate as only prior knowledge of the optimization parameter is known \cite{lu2015limitations} and thus no true input probability distribution needed for propagating the uncertainties is available. Therefore, here the gPC expansion still depends on realizations $\theta$ which comes at a greater computational cost but higher accuracy.
	\end{remark}

	Next, we aim to develop a numerical method for computing (\ref{g_optimization}).

	\subsubsection{Constrained Random Walk Markov Chain Monte Carlo}
	Markov Chain Monte Carlo (MCMC) \cite{brooks2011handbook} methods are frequently employed to effectively sample posterior distributions. \\ 
	The random walk MCMC (also known as Metropolis-Hastings-Algorithm) \cite{Kaipio2005} is one method to approximate the posterior probability distribution up to a normalizing constant by creating a Markov Chain of samples of the posterior. Here, candidate points are suggested based on a proposal distribution and then either accepted or rejected with a probability given by the acceptance ratio. In order to derive this ratio, the posterior at both the current sample and the proposed candidate is evaluated. This procedure is repeated, until enough samples are collected so that the Markov Chain converges to the stationary distribution. \\
	The standard algorithm is modified to incorporate the constraint by further using an indicator function to compute the ratio, following the approach described above from Gardner and Gelbert \cite{Gardner2014, Gelbert2015}. The pseudo-code is given in Algorithm \ref{cMCMC_alg} as an extension to the unconstrained algorithm presented in \cite{Kaipio2005}, and labeled constrained Random Walk Markov Chain Monte Carlo (cRW). \\
	
	In order to describe the method it is sufficient to state the transition kernel $K: \mathbb{R}^{N_{\theta}} \rightarrow \mathbb{R}^{\geq 0}$:
	\begin{equation}
		K(\theta, \theta^*)=\alpha(\theta, \theta^*) , \, \, \, \, \alpha(\theta, \theta^*) = \text{min}\left\{1, \frac{\Pi_{prior}(\theta^*) \cdot \Pi_{l}(\theta^*)}{\Pi_{prior}(\theta) \cdot \Pi_{l}(\theta)} \cdot \chi_\mathcal{S}(\theta^*)\right\},
		\label{kernel}
	\end{equation}
	where $\theta$ is a current sample and $\theta^*$ is a candidate point proposed by the random walk. The indicator function $\chi_\mathcal{S}$ and $\mathcal{S}$ are defined by eq. (\ref{subset}).
	\begin{theorem}
		The cRW transition kernel $K$ \textup{(\ref{kernel})} satisfies the detailed balance condition with target measure $\Pi^c_{post}(\theta)$ \textup{(\ref{g_optimization})} given a feasible initial sample $\theta$.
	\end{theorem}
	
	\begin{proof}
		The detailed balance equation reads as 
		\begin{equation}
			\Pi_{prior}(\theta^*) \cdot \Pi_l(f_1^{data}|\theta^*) \cdot K(\theta^*, \theta) = \Pi_{prior}(\theta) \cdot \Pi_l(f_1^{data}|\theta) \cdot K(\theta, \theta^*).
			\label{db}
		\end{equation}
		For any two samples, where $\chi(\theta^*)=\chi(\theta)$ the above statement is true. In case of $\chi(\theta^*)=\chi(\theta)=1$, (\ref{db}) simplifies to the standard detailed balance equation for random walk. If the indicator function for both samples is zero, so is the probability of either sampling one of them zero. \\
		In case $\chi(\theta^*)=0 \text{ and } \chi(\theta)=1$, the suggested point is not within the feasible region of the \textit{constrained} posterior $\Pi^c_{post}$, and will thus not be sampled. The rejection will lead to a repetition of the current sample $\theta$.
	\end{proof}

	\subsubsection{Alternative Solution Methods for Numerical Comparisons}
	Alternative Monte Carlo based sampling methods use the gradient information of the posterior to propose new candidate points with the goal to achieve higher acceptence rates. We consider two different gradient descent algorithms, the Hamiltonian Monte Carlo (HMC) \cite{wang2013adaptive} and, as an interacting particle system, the Stein Variational Gradient Descent (SVGD) introduced by Liu and Wang \cite{liu2016stein}. \\
	
	In this work, in order to treat the constrained posterior, we propose to introduce a penalty term. The gradient of the log-posterior is therefore supplemented by a modification $\nabla \mathcal{H}$, such that the gradient is given by
	\begin{subequations} 
		\begin{equation}
			\nabla \mathcal{G}(\theta)=\nabla \log{(\Pi_{prior}(\theta) \cdot \Pi_{l}(\theta))}+ \nabla \mathcal{H}(\theta),
			\label{Lagrangian}
		\end{equation}
		with
		\begin{equation}
			\nabla \mathcal{H}(\theta)= 
			\left\{ \begin{matrix}
				\delta &\text{ if } \theta \notin \mathcal{S} \\
				0 &\text{ otherwise} \end{matrix} \right.,
			\label{heaviside}
		\end{equation}
		where $\delta$ is a hyperparameter used to avoid sampling in the unfeasible area.
	\end{subequations}
	The pseudo-codes for the extended methods are given in Algorithm \ref{cHMC_alg} and Algorithm \ref{cSVGD_alg} as extensions of the unconstrained algorithms presented in \cite{wang2013adaptive} and \cite{liu2016stein}. They are labeled constrained Hamiltonian Monte Carlo (cHMC) and constrained Stein Variational Gradient Descent (cSVGD), respectively. \\
	It should be noted the detailed balance equation is not fulfilled for the cHMC algorithm due to the modification.
	
	\section{Computational Results}
	\subsection{Space-homogeneous transpiration cooling}
	For model 1, flow through the porous medium is simulated. Two stochastic input parameters are considered, which are both assumed to be normally distributed with $\phi \sim \mathcal{N}(\phi_0, \sigma_\phi^2)$ and $q \sim \mathcal{N}(q_0, \sigma_q^2)$. The prior of the Reynolds number is assumed to be a Gaussian prior with $\mathcal{N}(\mu_p, \sigma_p^2)$. While the pressure is observable, the temperature is constrained. \\
	
	For the cHMC and cSVGD, the use of the log-posterior avoids overflow and underflow in the computation of the ratio for the acceptance rate \cite{Somersalo2007}. Furthermore, the particles are following a smoothed gradient. For this model, the gradient can be derived analytically and is given by 
	
	\begin{equation}
		\begin{split}
			&\nabla_{Re} \log{\left(\Pi_{prior}\left(Re\right) \, \frac{1}{N} \prod\limits_{i = 0}^{N} \Pi_l\left(p_i(\phi_0)|Re\right)\right)} \\
			& = - \frac{(Re-\mu)}{\sigma_{p}^2} + \left( \nabla_{Re} p(Re) \right) \frac{1}{N \sigma_l^2} \sum\limits_{i=0}^{N} \left(p_i^{data} - p(Re) \right), \\
		\end{split}
		\label{gradient1}
	\end{equation}
	where $\left( \nabla_{Re} p(Re) \right)$ is evaluated using first-order finite differences. As the derivation with respect to a single input, the Reynolds number, is considered, two evaluations of the ODE are needed for finite differences. Therefore, using an adjoint gradient would not be more efficient to use. \\
	
	Figure \ref{Hist_M1} shows the sampled constrained posterior probability distribution using cRW, cHMC and cSVGD. The results reproduce the true posterior. It can be seen that for lower Reynold numbers the temperature constraint is not fulfilled as not enough coolant is injected to decrease the temperature at the interface with the desired probability of $95 \% $. \\
	While the cRW is using a hard constraint formulation to sample the constrained posterior, the cSVGD and cMHA violate the constraint requirement based on the choice of the penalization parameter $\delta$. Therefore, if a weak penalization is chosen in (\ref{heaviside}), not all samples might be within the feasible region. On the other hand, the modification alters the shape of the posterior and if $\delta$ is be chosen to be very large, it will affect the sampling at the edge of the feasible region. The non-feasible samples may be removed by checking the constraint in a post-processing step, however this new distribution does not resemble the Markov Chain.
	
	{\begin{figure}[H]
			\centering
			\begin{minipage}[t]{0.31\linewidth}
				\includegraphics[width=1.1\linewidth]{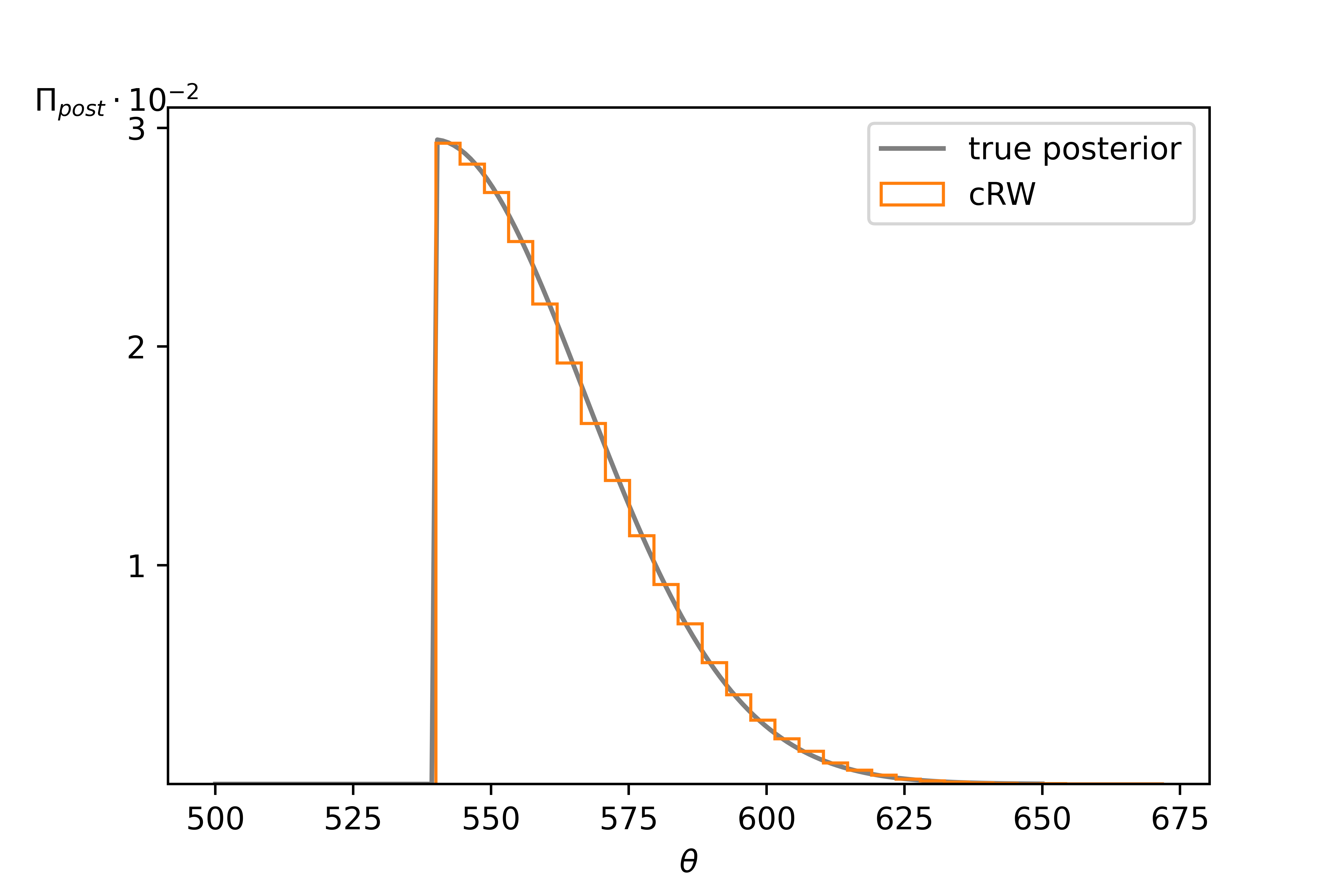}
			\end{minipage}
			\begin{minipage}[t]{0.31\linewidth}
				\centering
				\includegraphics[width=1.1\linewidth]{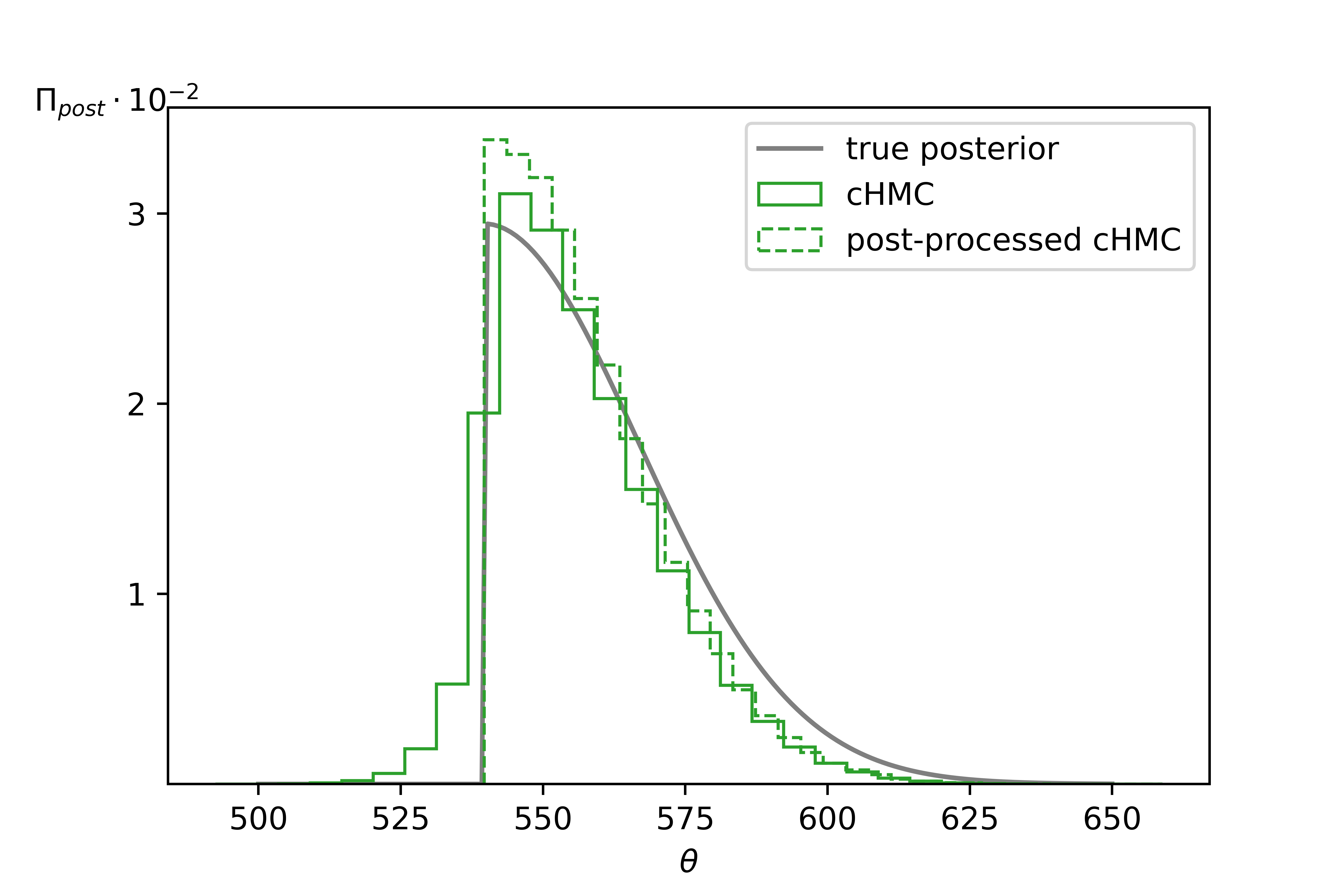}
			\end{minipage}	
			\begin{minipage}[t]{0.31\linewidth}
				\centering
				\includegraphics[width=1.1\linewidth]{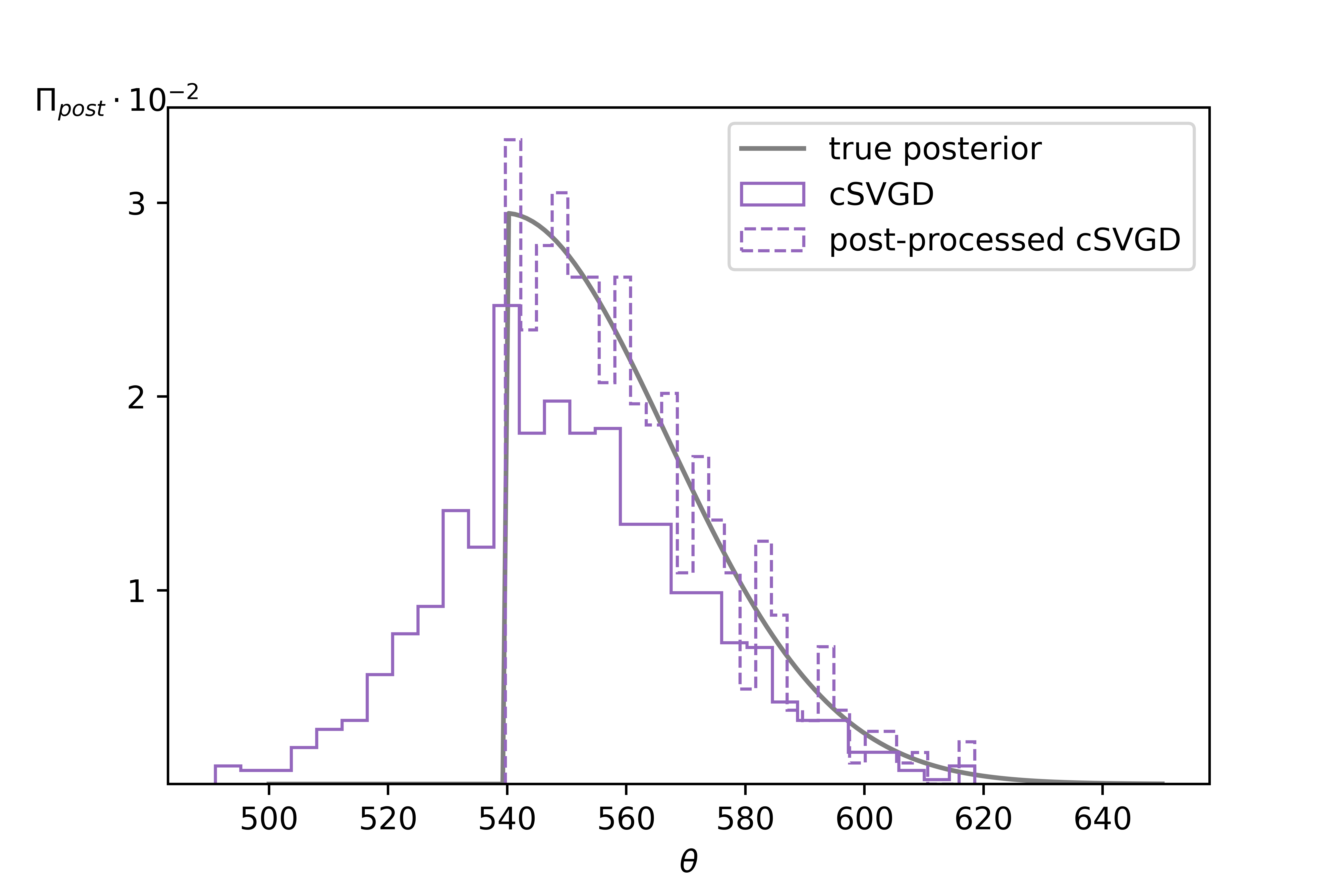}
			\end{minipage}	
			\caption{Histograms Model 1: The results of the cRW, the cHMC and the cSVGD can be seen from left to right. Additionally to the Markov Chain, post-processed versions are shown for the cHMC and cSVGSD. Here, unfeasible samples have been removed from the corresponding Markov Chains. For comparison, the true posterior (gray line) is plotted as well. It can be seen that for a Reynolds number less than $540$, the temperature constraint is not fulfilled. The set $\mathcal{S}$ is $=\{ \theta \geq 540\}$. }
			\label{Hist_M1}
	\end{figure} }

			The temperature is increasing as the mass flow of the coolant and thus the Reynold number increases. Therefore, Figure \ref{Hist_M1} shows that in this scenario, the chance constraint is concomitantly a box constraint. In such a case, a projected gradient method can be used instead of the modification. This circumvents the downsides of the modification as described above, where the shape of the posterior is altered, and samples without the feasible region might be removed in a post-processing step. However, in the general case given in (\ref{g_optimization}), the secondary condition does not necessarily resemble a box constraint and the use of a projected gradient method might not be possible. \\
			The projected gradient method has been combined with the Stein Variational Gradient Descent for demonstration purposes. In case where $\mathcal{S}=\{ \theta \geq \theta_{crit} \}$, we can reformulate (\ref{g_optimization}) as
			
			\begin{equation}
				\Pi_{prior}(\theta) \cdot \Pi_l(f_1^{data}|\theta) \text{ s.t. } \theta \geq \theta_{crit}.
			\end{equation}
			For every particle $\{\theta_k^j\}_{k=1}^m$ in the $j$th generation of the algorithm with a set of $m$ particles, the known SVGD proposal step is performed. The gradient of the unconstrained posterior at $\theta_k^j$ is then projected back onto the feasible region by a projection step
			\begin{equation}
				d_k^j=\text{Pr}(\theta_k^j + \nabla_\theta \log{\Pi_{prior}(\theta_k^j) \cdot \Pi_l(\theta_k^j)} ) - \theta_{k}^j. 
			\end{equation}	 
			Finally, the particles are then updated by $\theta_{k}^{j+1} = \theta_{k}^{j} +s^j \cdot d_k^j $, where $s$ is the current SVGD step size. \\ 
			
			When $\mathcal{S} \subset \mathbb{R}^{N_\theta}$ is the set such that $\theta \geq \theta_{crit}$, then the projection operator $\text{Pr}: \mathbb{R}^{N_\theta} \rightarrow \mathcal{S}$ \cite{Burke} yields the closest feasible point to $y := \theta_k^j + \nabla_\theta \log{\Pi_{post}(\theta_k^j)}$ by solving
			\begin{equation}
				\min_{\theta \in \mathcal{S}}{\frac{1}{2} || \theta - y||_2^2}.
				\label{proj}
			\end{equation}
			In this case, for any $y \notin \mathcal{S}$, (\ref{proj}) returns $\theta = \theta_{crit}$. The results for the projected SVGD method can be seen in Figure \ref{proj_SVGD}.
			
			{\begin{figure}[H]
					\centering
					\includegraphics[width=1\linewidth]{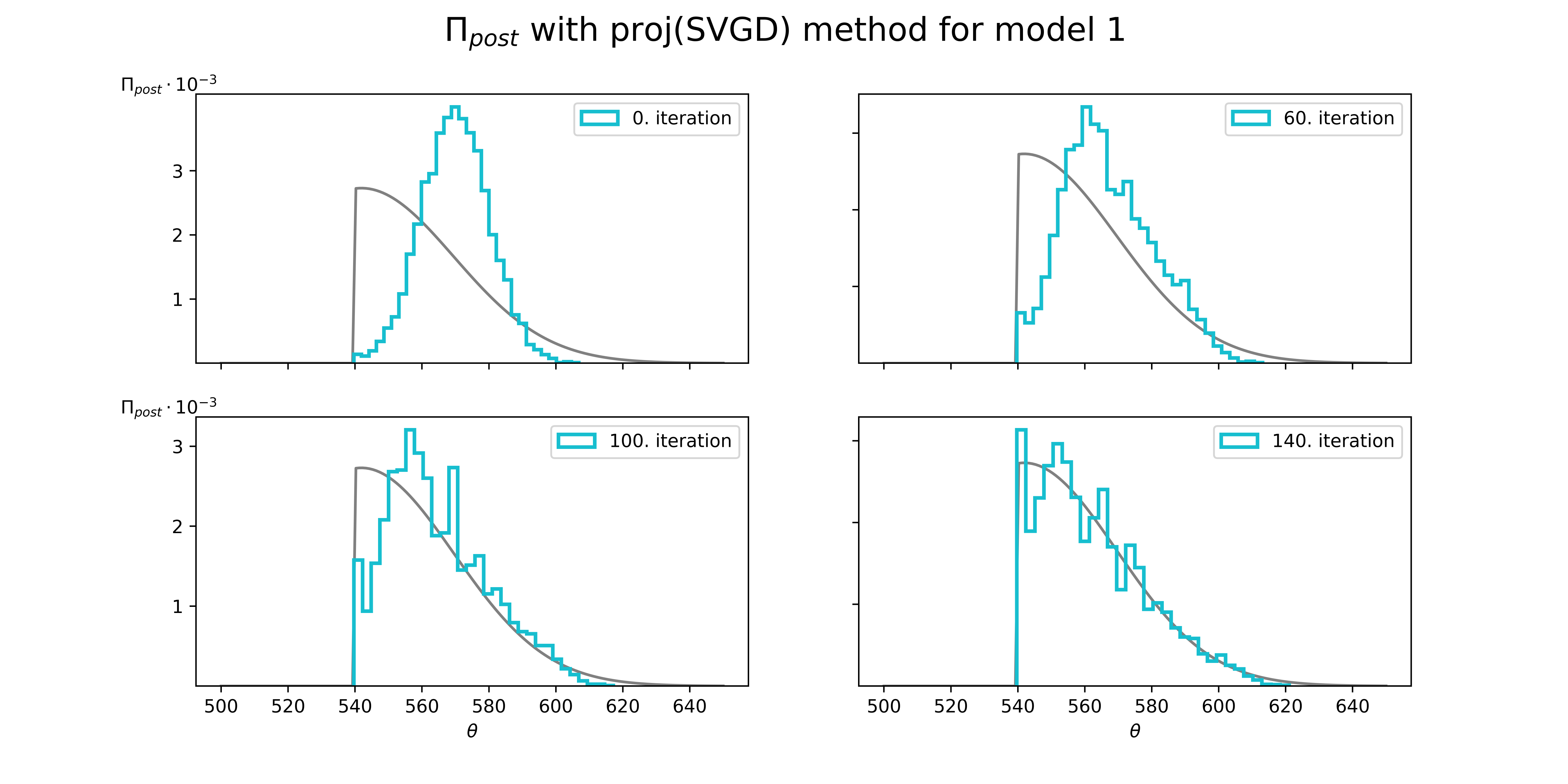}
					\caption{Histogram for the projected SVGD method. The evolution of the particle distribution can be seen from the initial distribution in the upper left plot over the $60$th and $100$th generation to the final distribution in the lower right plot. With this method, the constraint is not violated as the gradient is projected back onto the feasible region in every step. Again, the gray line shows the true posterior for comparison.}
					\label{proj_SVGD}
			\end{figure} }

			\subsection{Space-dependent transpiration cooling}
			
			For model 2, four different sections  are considered over the length of the interface, with two walls and two porous medium flows as it can be seen in Figure \ref{plot_constraint1}. As the uncertainty of the heat flux enters the temperature equations (\ref{2a}) and (\ref{2b}), the initial condition of (\ref{hd}) is stochastic for the porous medium sections. We assume a uniformly distributed prior with $Re \sim \mathcal{U}(300, 1000)$. As two flows with different porosities are considered, two separate Likelihood-functions are taken into account. The temperature constraint states that a temperature of $380 K$ should not be exceeded with a probability of $80 \%$. \\
			
			For reference, the posterior was computed for realizations of the Reynolds number and the stochastic variable $q$. Figure \ref{plot_contour} shows the contour plot. \\
			
			{\begin{figure}[H]
					\centering
					\includegraphics[width=0.8\linewidth]{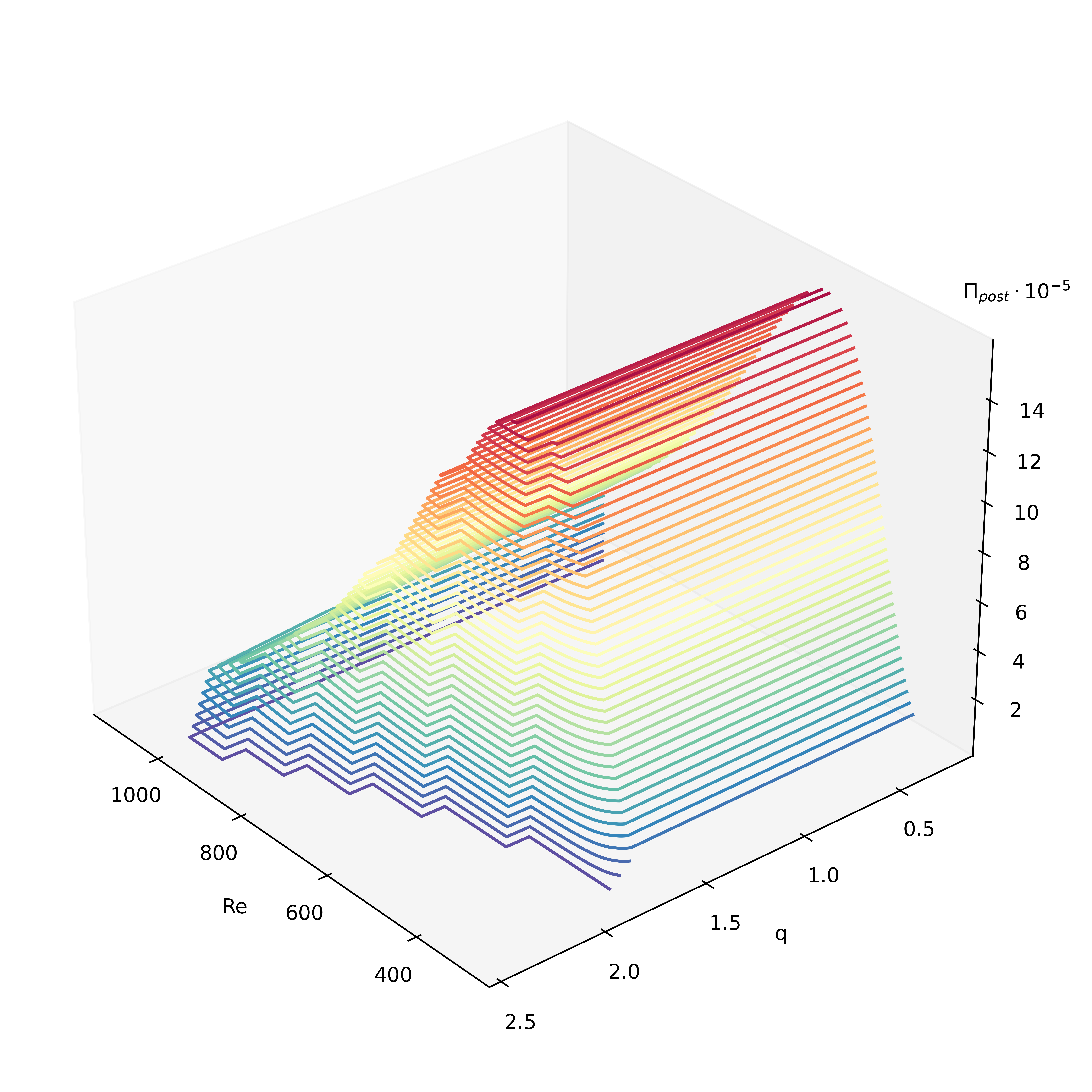}
					\caption{Reference Posterior: Computing the constrained posterior for different realizations of the optimization parameter and the stochastic variable, it can be seen that for lower heat fluxes and higher Reynolds numbers the constraint is fulfilled, as these combinations produce lower temperatures.  Here, the constraint is checked pointwise. The gPC-constraint in the sampling strategies below, however, ensures to erase any Reynolds numbers for which the probability is too high that the normally distributed heat flux results in exceeding the maximum temperature.}
					\label{plot_contour}
			\end{figure} }
			
			In Figure \ref{plot_constraint1}, the stochastic simulation is shown at the start time and $t=t_c$, on which the constraint is imposed for a Reynolds number $Re=405$.  It shows that the maximum temperature is clearly exceeded. In Figure \ref{plot_constraint2}  the expected value of the sampled Markov Chains is used for the same simulation. Here, the temperature constraint is fulfilled. The posterior probability distributions are shown in Figure \ref{Hist_M2}. For the gradient-based methods, the gradient in this case can be derived as 
			
			\begin{equation}
				\begin{split}
					&\nabla_{Re} \log{\left(\Pi_{prior}\left(Re\right) \, \frac{1}{N} \prod\limits_{i = 0}^{N} \Pi_l\left(p_i|Re\right)\right)} \\
					& = \left( \nabla_{Re} p(Re) \right) \cdot \left( \frac{1}{N \sigma_{0}^2} \sum\limits_{i=0}^{N}  \left(p_i^{data_0} - p(Re, \phi_0) \right) +\frac{1}{N \sigma_{1}^2} \sum\limits_{i=0}^{N}  \left(p_i^{data_1} - p(Re, \phi_1) \right) \right). 
				\end{split}
				\label{gradient2}
			\end{equation}
			The uniformly distributed prior with $Re \sim \mathcal{U}(300, 1000)$ is not differentiable at $Re=400$ and $Re=1000$. For the SVGD - that entirely relies on the gradient and the computation of the posterior probability at the particle positions itself is not part of the algorithm - an additional penalization in the same manner as for the inequality is imposed for $Re>1000$ to avoid sampling in that region. \\
			For the HMC, as $\log{ \Pi_{prior}(\theta^*) \cdot \Pi_{l}(\theta^*)}$ at the proposal $\theta^*$ is computed for the acceptance ratio, the prior is set to be $0 < \epsilon << 1$ for $Re < 400$ and $Re  > 1000$. In practice this leads to a rejection of proposal candidates for $Re \notin (300, 1000)$ using the cHMC, as the probability of the posterior is nearly zero due to the extremely small $\epsilon$. 
			
			{\begin{figure}[H]
					\centering
					\includegraphics[width=0.75\linewidth]{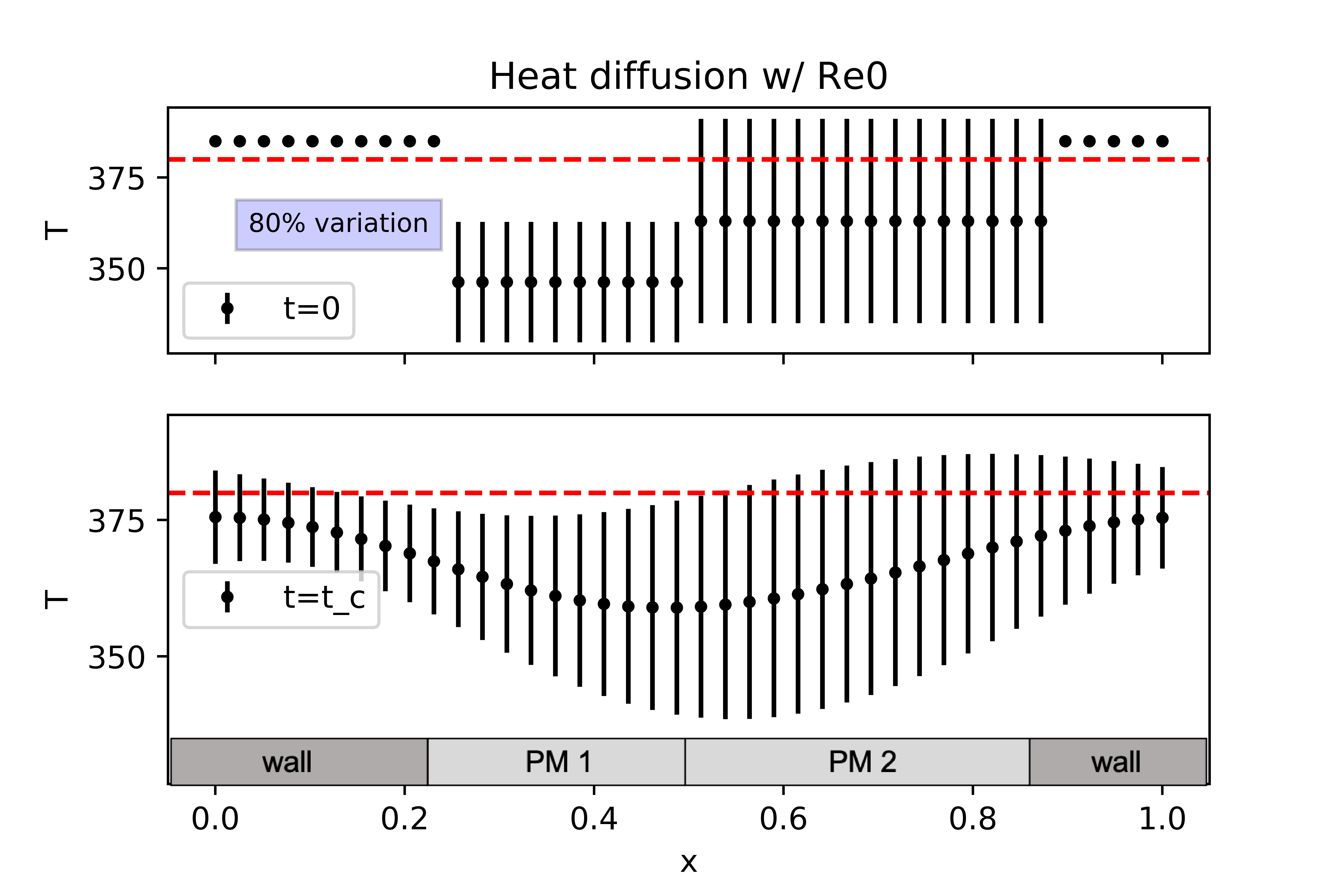}
					\caption{Constraint satisfaction: This figure shows the temperature distribution at time $t=0$ and $t=t_c$. The initial parameter value for the Reynolds number $Re=405$ of the deterministic simulation (see Appendix 1) was used to simulate the stochastic model. It is obvious that not only at $t=0$ the constraint is not fulfilled, but also at $t=t_c$ the temperature still exceeds the threshold with a probability greater than $80\%$. }
					\label{plot_constraint1}
			\end{figure} }
			
			{\begin{figure}[H]
					\centering
					\includegraphics[width=0.75\linewidth]{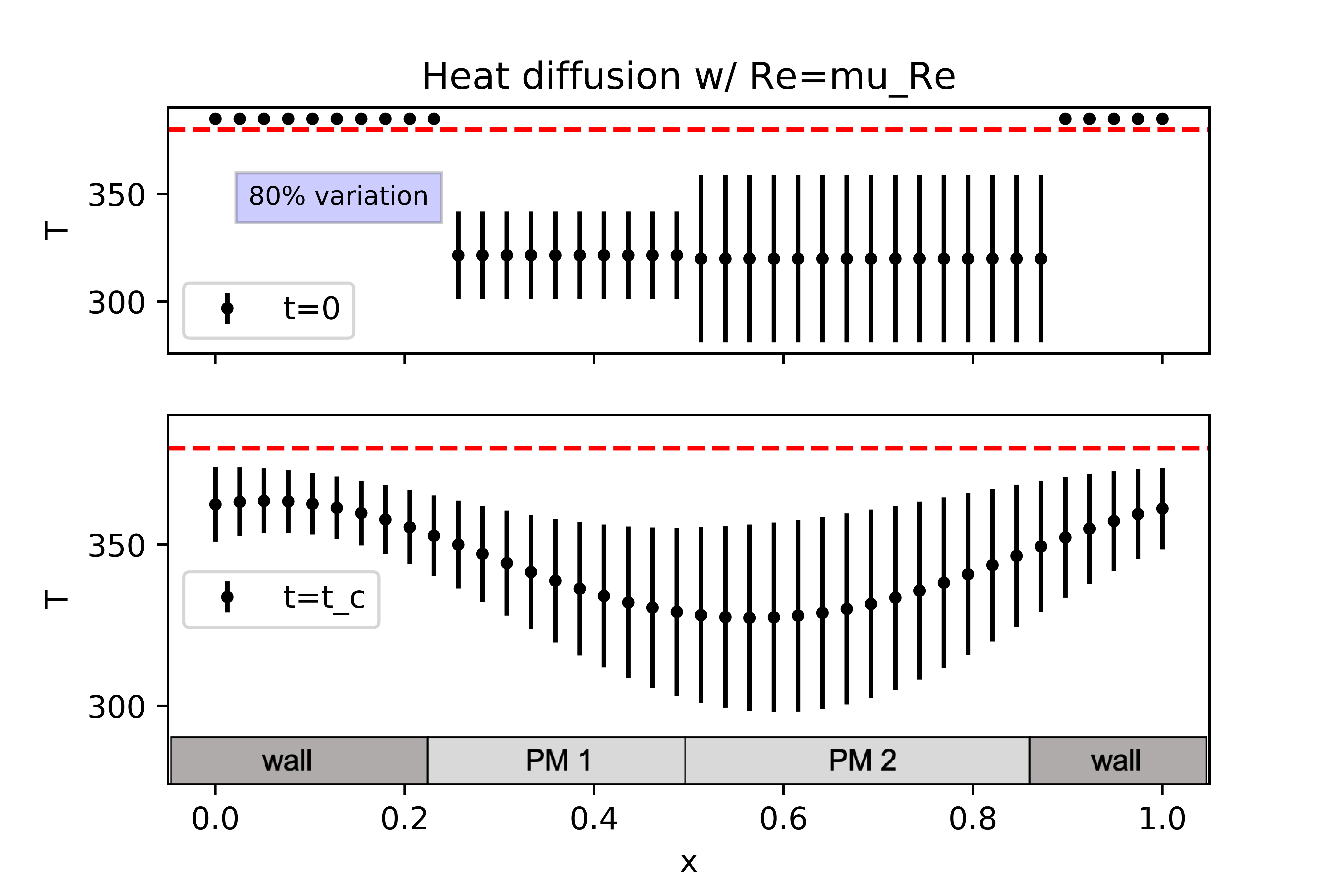}
					\caption{Constraint satisfaction: This figure shows the same design as Figure \ref{plot_constraint1} with the difference that here, the expected value of the constrained Monte Carlo methods is used to simulate the stochastic model. It can be seen that at $t=0$ the wall temperature is above the maximum temperature while enough coolant was injected to already achieve sufficient cooling right at the outflow of the porous medium. At time $t=t_c$, the temperature is below the maximum temperature everywhere with a probability of $80\%$.  }
					\label{plot_constraint2}
			\end{figure} }
			
			{ \begin{figure}[H]
					\centering
					\begin{minipage}[t]{0.31\linewidth}
						\includegraphics[width=1.1\linewidth]{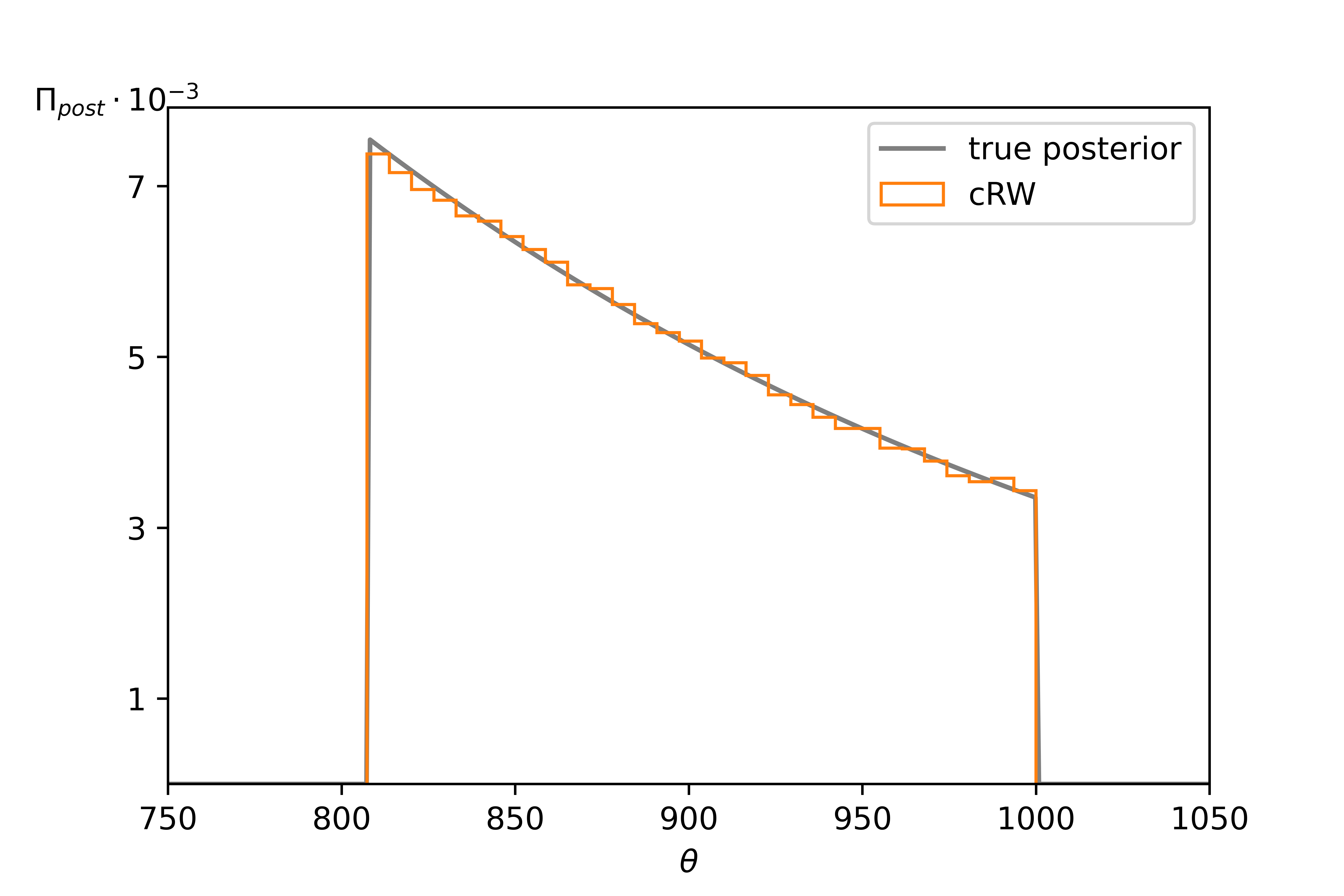}
					\end{minipage}
					\begin{minipage}[t]{0.31\linewidth}
						\includegraphics[width=1.1\linewidth]{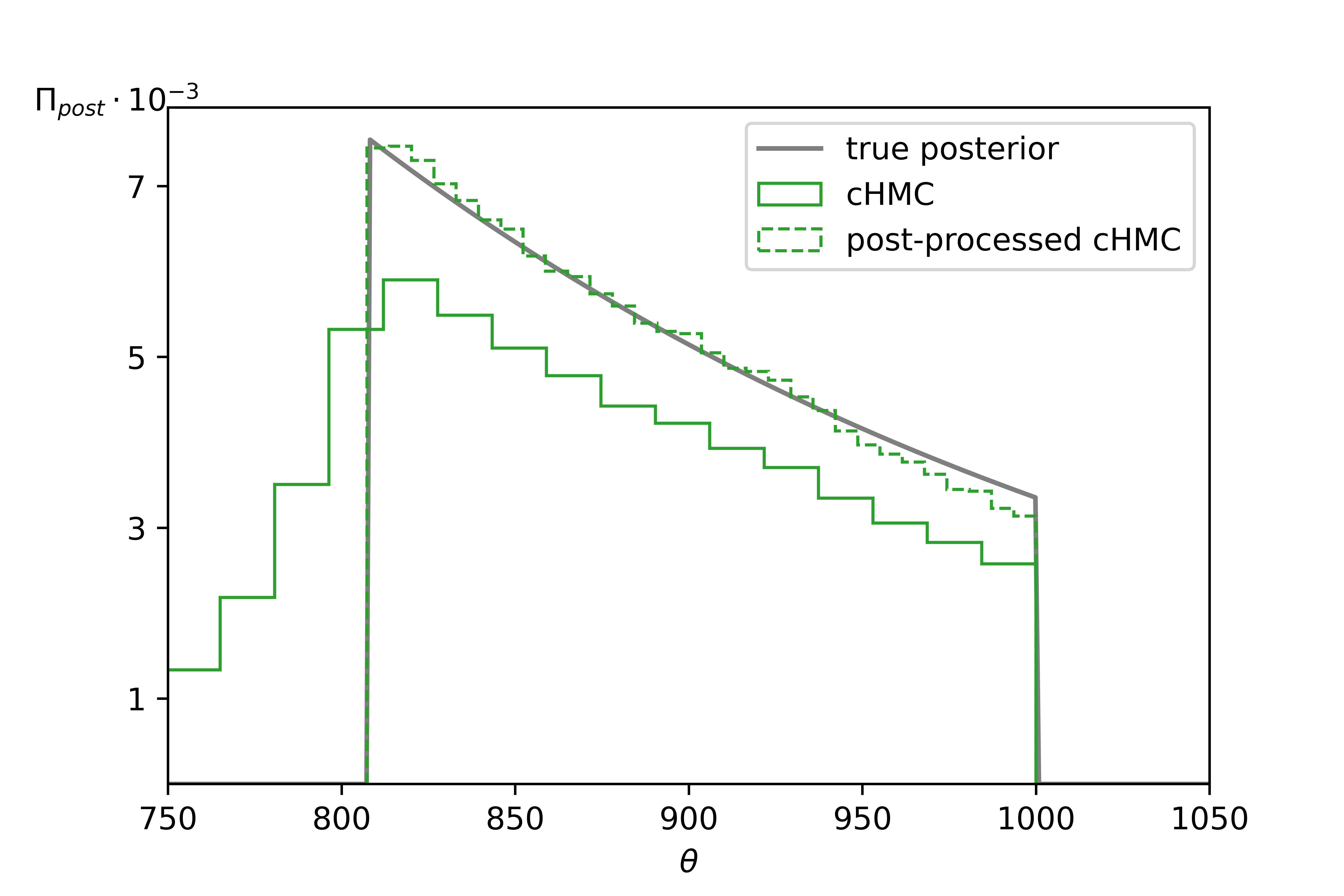}
					\end{minipage}	
					\begin{minipage}[t]{0.31\linewidth}
						\includegraphics[width=1.1\linewidth]{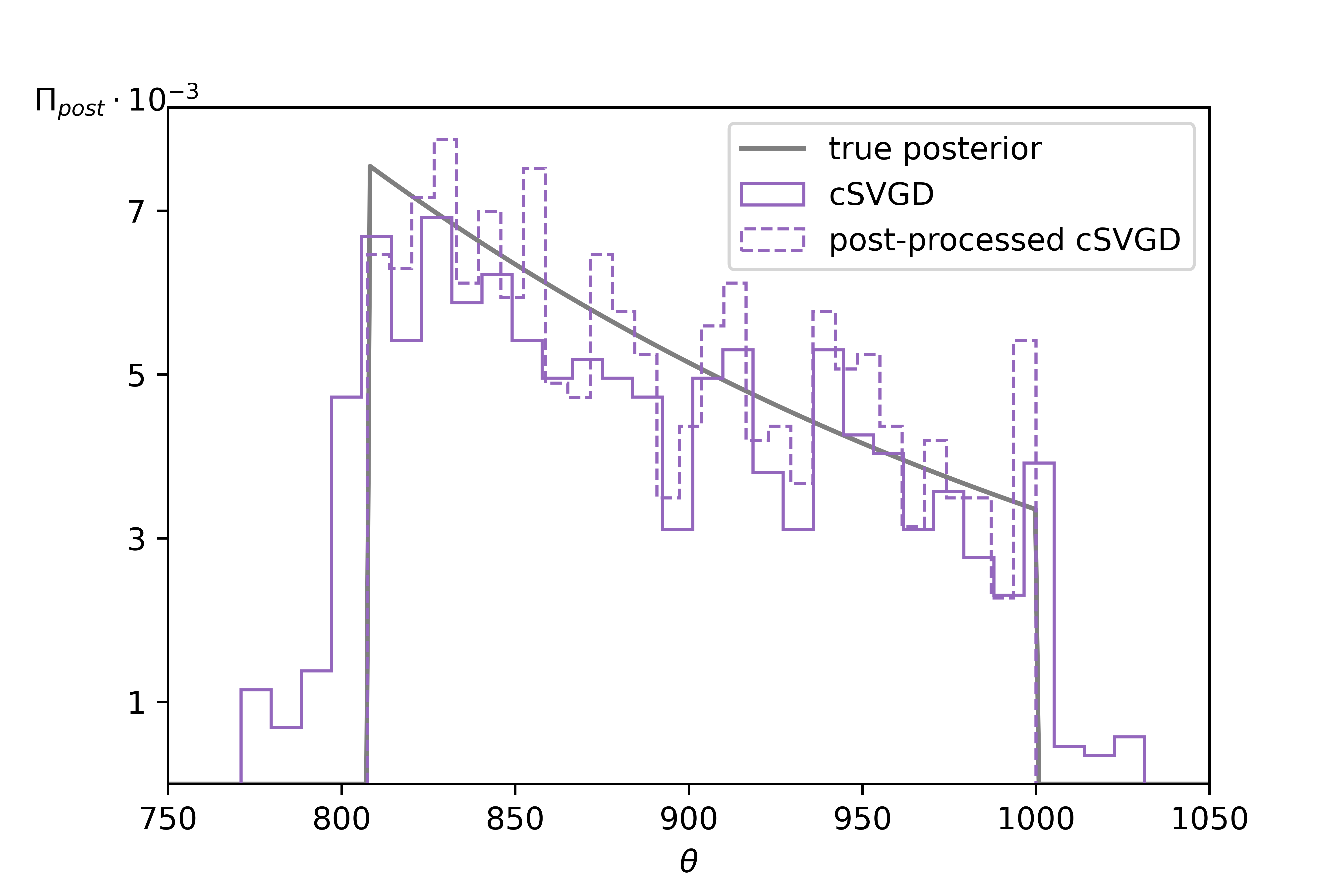}
					\end{minipage}
					\caption{Histograms Model 2. The histograms show results of the cRW, cHMC, and cSVGD, respectively. For both the cHMC and the cSVGD post-processed versions for which unfeasible samples are excluded are additionally plotted. }
					\label{Hist_M2}
			\end{figure} }
			
			\subsection{Space-dependent transpiration cooling with high-dimensional uncertain input}
			
			The space-dependent transpiration cooling model can be extended to the case with point wise uncertain heat flux over the length of the porous probe. For that, 60 independent RVs are all assumed to be normally distributed with known mean and variance according to simulation data of the heat flux from a deterministic coupled simulation setting with a hot gas flow developed in \cite{Dahmen2015} as it can be seen in Figure \ref{pointwise_input_uncertainty}. The uncertainties are propagated through the pores and form the initial conditions for the heat diffusion at the interface. \\
			
			The results for the cRW can be seen in Figure \ref{hist_M3}. 
			
			{\begin{figure}[H]
					\centering
					\includegraphics[width=0.8\linewidth]{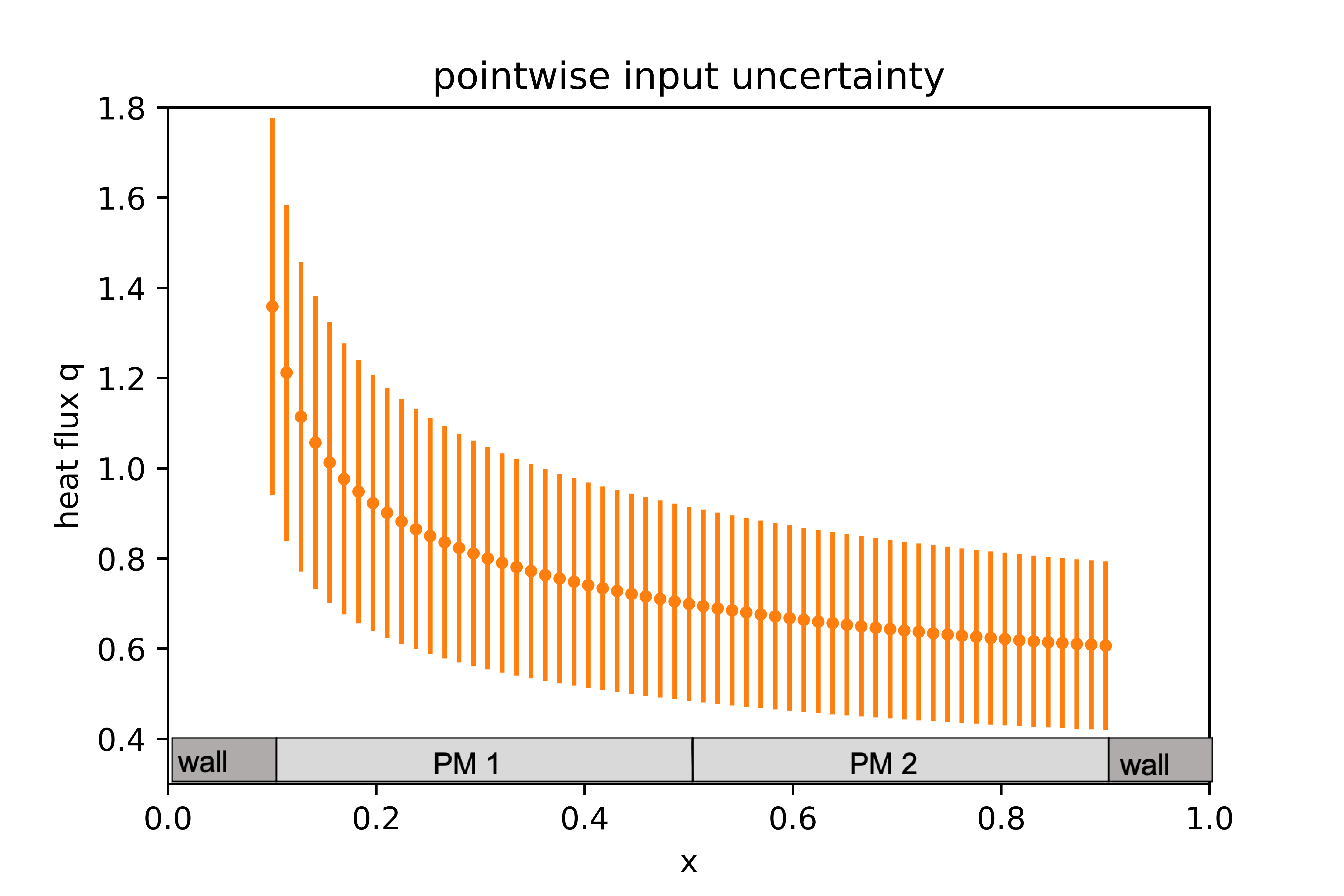}
					\caption{This figure shows the pointwise input uncertainty of the heat flux over the interface. As the RVs are assumed to be normally distributed, mean values and corresponding standard deviations are featured. }
					\label{pointwise_input_uncertainty}
			\end{figure} }
			
			{\begin{figure}[H]
					\centering
					\includegraphics[width=0.8\linewidth]{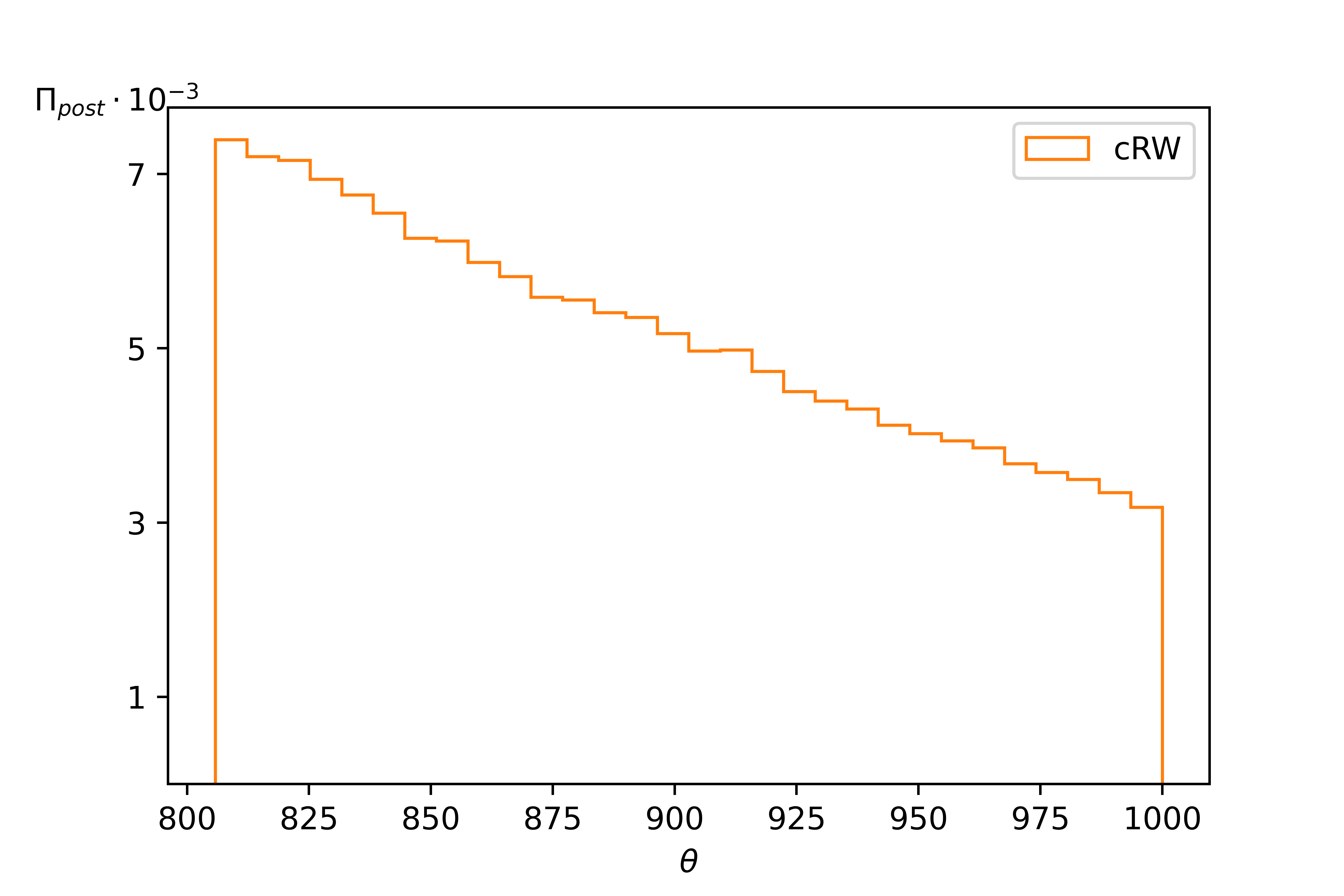}
					\caption{The histogram of the cRW method using 500000 samples for model 3 is shown. Here, $\mathcal{S} = \{ 806 \leq \theta \leq 1000\}$. }
					\label{hist_M3}
			\end{figure} }
			
			\subsection{Evaluation of the Monte Carlo based methods}
			The accuracy of the proposed methods depends on both the accuracy of the gPC expansion and the accuracy of the MCMC sampling strategies. The accuracy of the gPC expansion relate on the truncation order $P$ which needs to be high enough so that the function $f_2^{P}$ evaluated over $\xi$ represents the true function $f_2(\boldsymbol{\xi}; \theta)$. \\
			
			To compare the proposed methods, however, the accuracy of the MCMC methods should be considered. Estimates of the true posterior distributions have been derived for the first and second model, where the posterior was evaluated at selected points. Therefore, a measure of accuracy is the L2 error between the "true" posterior and the bars of the histograms. An example of the graphical evaluation can be seen in Figure \ref{example_L2}.
			
			{\begin{figure}[H]
					\centering
					\includegraphics[width=0.8\linewidth]{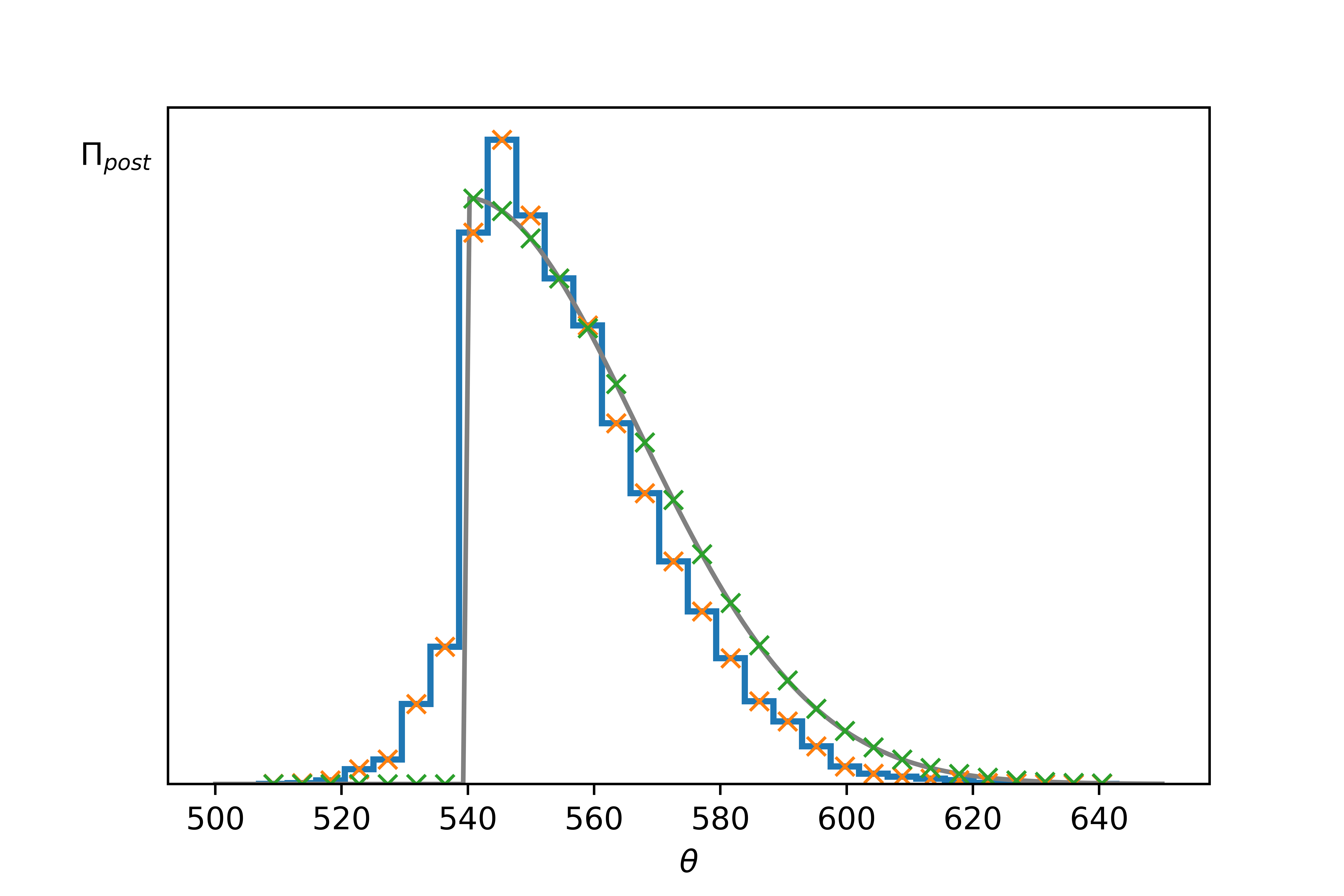}
					\caption{The relative $L^2$ error is shown. The histogram is the result of a simulation with $25000$ samples using the cHMC method for model 1. The bars of the histogram are shown with orange crosses at the middle of each bar. The corresponding point of the (estimated) true posterior is derived using interpolation (green crosses). The points are then used to compute the relative $L2$ error.}
					\label{example_L2}
			\end{figure} }
			
			For comparability between the methods, shared measurement points are needed. Here, we choose the number of samples ($N$) as measurement points. This compares the end-result of the algorithms which are equivalent long Markov Chains. For this measure, the number of samples for the particle-based SVGD are computed as the number of particles times the number of steps per particle as this the closest comparison to the overall number of samples. However, this in addition to the unused possible parallelization might unfavor the SVGD methods. \\
			While the quality of samples might profit from the gradient-based proposal step of the cSVGD and cHMC in contrast to the random choice of the cRW, the modification still allows the cSVGD and cHMC to sample without the feasible region which is expected to negatively influence the L2 norm at the same time. \\
			Furthermore, the evaluation of the gradient is time-consuming. It is likely that an evaluation of the gradient is dominant time-wise over the quality of the proposed sample. This is especially true in the case of cHMC, where the gradient is computed on top of the computation of the unconstrained posterior, instead of replacing this step (as in cSVGD). Therefore, the CPU time is tracked. For the pSVGD and cSVGD the parallelization is not used such that the CPU times are comparable to the others. All instances were solved with a computer with Intel(R) Xeon 8160 CPUs with 2.10 Ghz. However, we are mainly interested in the relative CPU times between the different algorithms rather than absolute CPU times. \\
			All in all, the comparison yields in the trade-off between costly, but well-selected candidates and fast, but random suggestions in the constrained case. \\
			The results are shown for model 1 and model 2 in Figure \ref{sample}. While all methods converge for very high number of samples, the cRW shows the smallest error. For the cSVGD and cHMC the post-processed chains, where all non-feasible samples are removed, are shown additionally. It can be seen that the performance does indeed increase when the non-feasible samples are removed. \\
			The SVGD and projected SVGD require a higher number of repetitions. Their performance might be underrepresented here, as the number of particles and the number of steps need to be balanced such that they compare to the number of suggested candidates. They do perform well when using a higher number of particles, which could be iterated in parallel. However, as they do not include a rejection step, this might be the best method to compare them to the other constrained methods. \\
			A comparison of the CPU times shows that the cHMC takes considerably longer than the other methods. Therefore, the quality gain from the gradient-information in the HMC does not automatically make it a favorable option.

			{\begin{figure}[H]
					\centering
					\begin{minipage}[t]{0.49\linewidth}
						\includegraphics[width=1.1\linewidth]{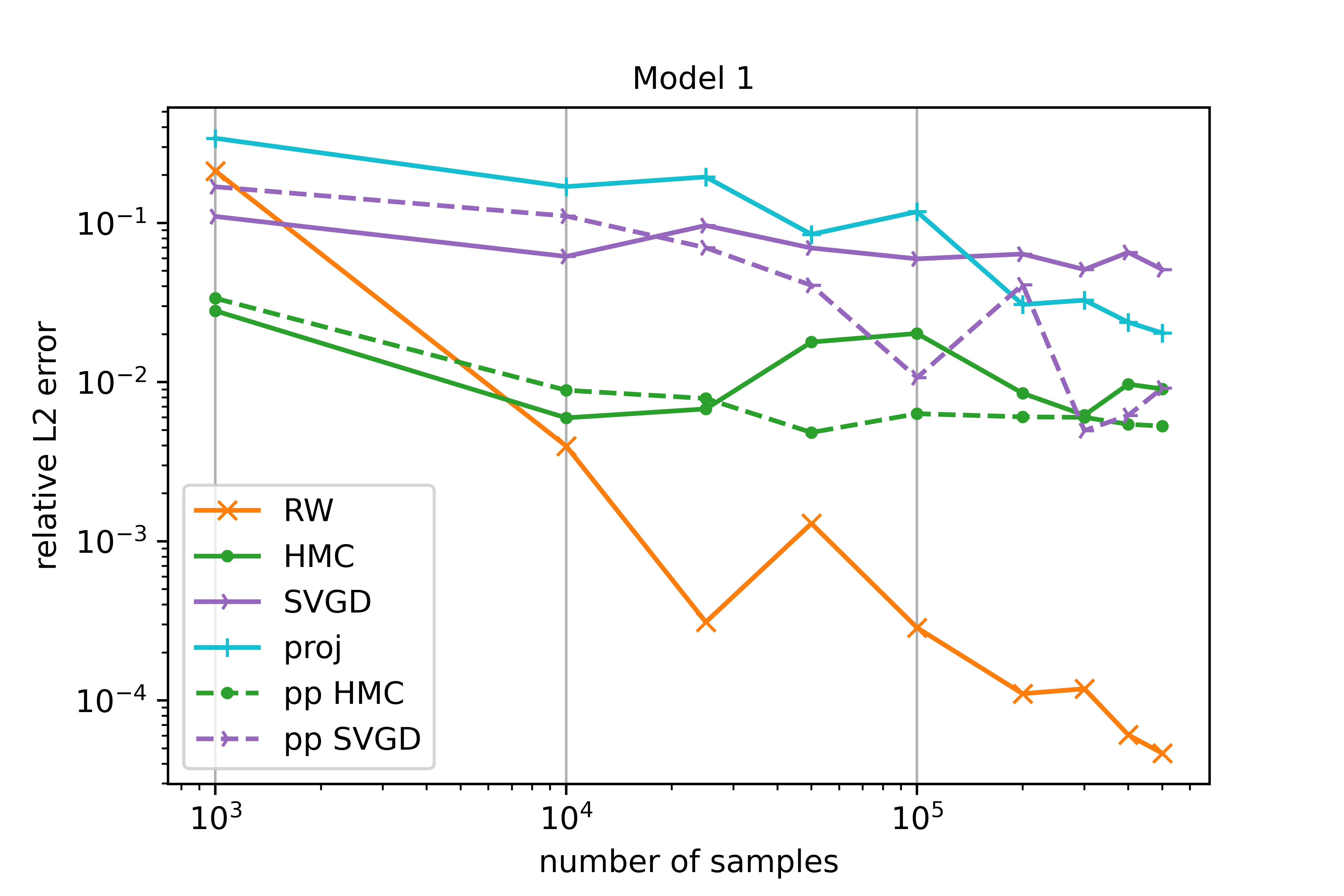}
					\end{minipage}
					\begin{minipage}[t]{0.49\linewidth}
						\includegraphics[width=1.1\linewidth]{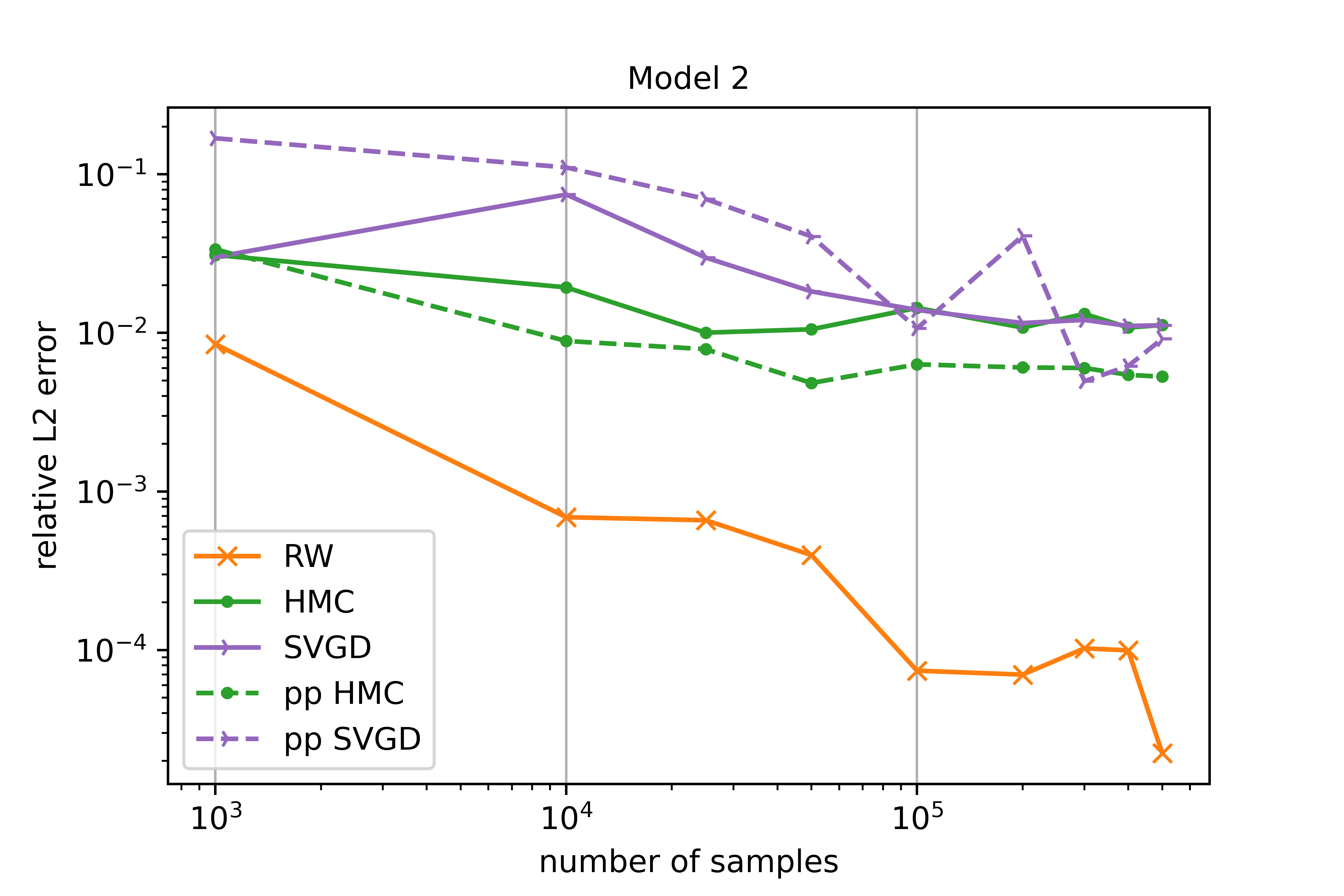}
					\end{minipage}	
					\begin{minipage}[t]{0.49\linewidth}
						\includegraphics[width=1.1\linewidth]{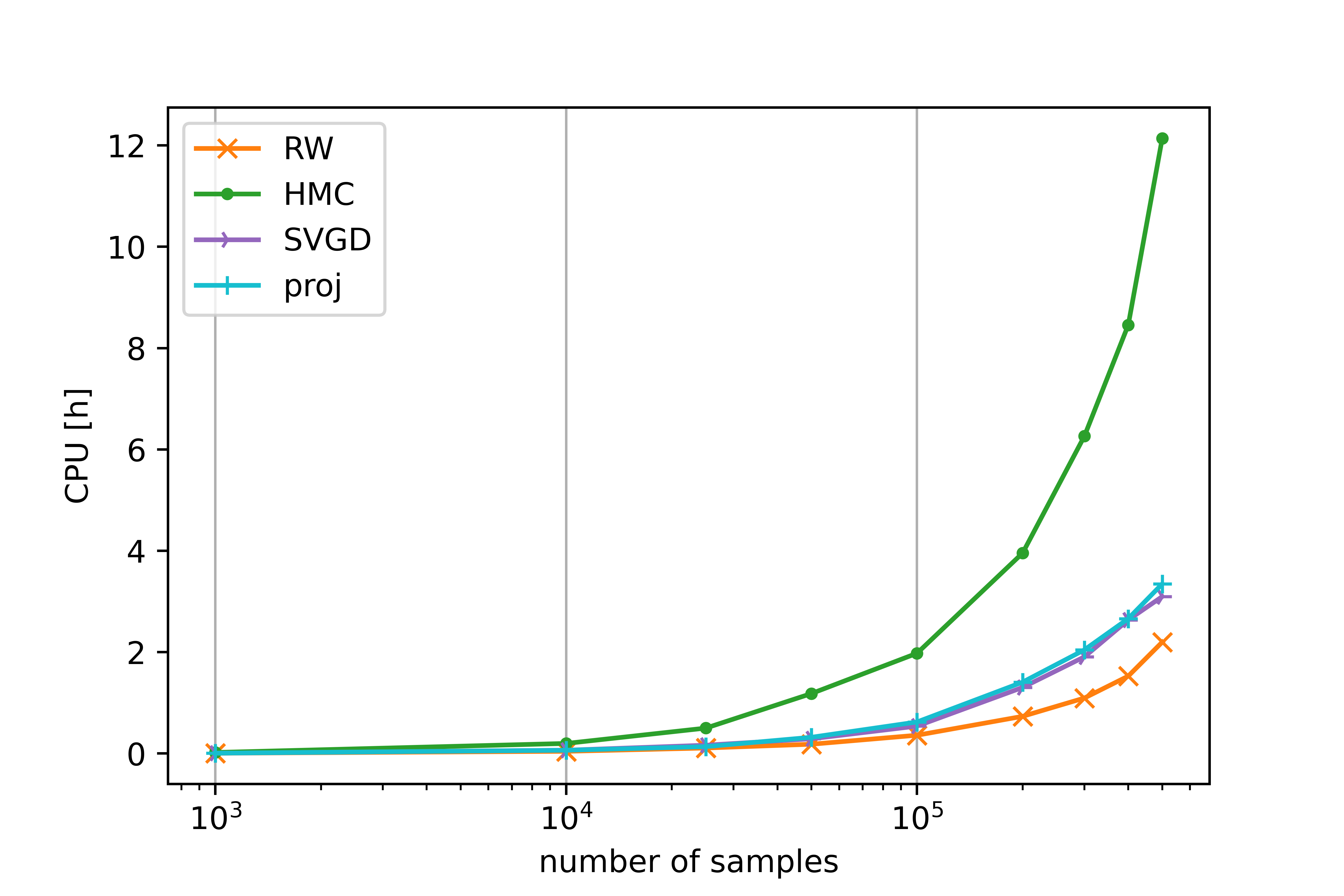}
					\end{minipage}
					\begin{minipage}[t]{0.49\linewidth}
						\includegraphics[width=1.1\linewidth]{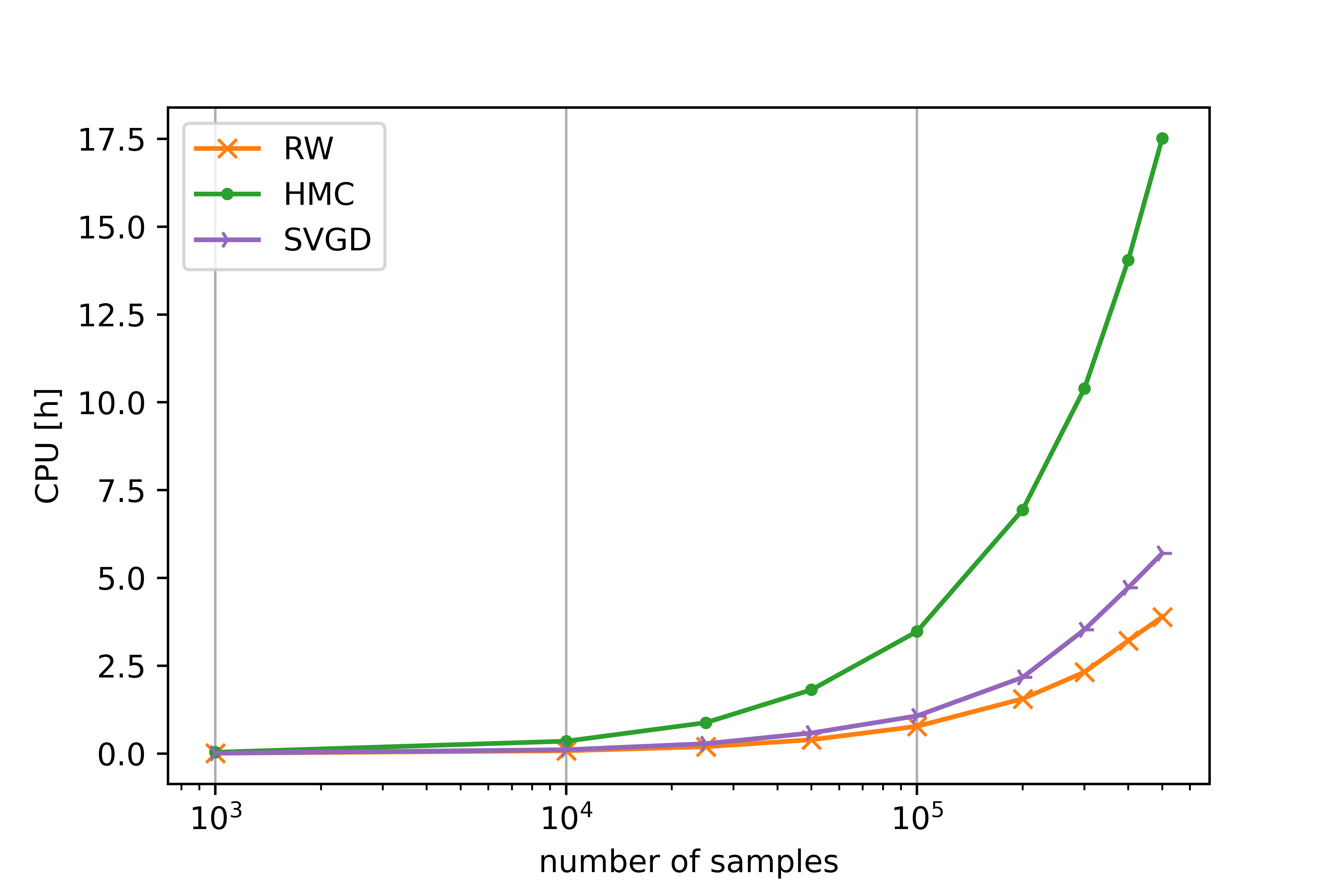}
					\end{minipage}	
					\caption{This figure shows the relative error over the number of samples for model 1 on the left and model 2 on the right. The associated CPU times are shown for each method, respectively.}
					\label{sample}
			\end{figure} }
			
			The choice of hyperparameters influences the performance of the algorithms. For cRW, this is the variance of the random walk. For the cHMC, the momentum, step size, number of LeapFrog steps and the penalization parameter $\delta$ need to be tuned. The Python SVGD algorithm provided by Li already contains a step size adaption methods, such only the initial step size, and not a step size in each iteration round of the particles, needs to be chosen. Furthermore, for the constrained case, the penalization parameter needs to be adapted. For the second model there are two penalizations, as the uniform prior imposes an additional constraint for $\theta > 1000$, which needs to be integrated into the SVGD framework. \\ It is difficult to avoid this source of bias on the performance, as the hyperparameters can only be adjusted as "optimal" as possible for each method, however their fitness is not comparable between the different algorithms. Moreover, as no specific tuning recommendations exist for constrained MCMC methods, the choice of the penalization parameter remains subjective, where the only qualitative criteria is the desired strictness of the constraint fulfillment. This might be highlighted as a key advantage of the cRW, where this is not an issue. \\ Furthermore, common hyperparameter adaption methods can of course be used, however their interference with the constraint are not necessarily clear. All in all, the cRW and the cSVGD with one and two hyperparameters, respectively, seem to be favorable over the cHMC with respect to the hyperparameter selection. \\
			For this analysis, the hyperparameters were tuned for each method and model using the knowledge of the "true" posterior. For the cHMC and cSVGD the distribution of the post-processed samples rather than the whole Markov Chain was used to choose optimal hyperparameters. However, it should be stressed that the performances of the algorithms are highly influenced by the parameter choices, and that this puts a unresolvable bias on the direct comparison. \\
		
			Overall, the cRW shows the best results taking into account the overall performance, quality stability of the proposed candidates with respect to the constrained case and CPU time. As only one hyperparameter is in need of tuning, the required model-specific adjusting of the method is small. The cSVGD, even if not proven best performance comparatively, also has strong potential when used with many particles in parallel which is especially appealing for models with high computational effort. The cHMC, even though providing good results when tuned correctly, has proven to be both very time-costly and hyperparameter sensitive. Therefore, further diagnostics for the cRW are given in the following. \\
			
			For the general case, in which the true posterior is unknown, different convergence diagnostics exist to assess the convergence of the Markov Chain to the stationary distribution. Often these approaches concentrate on the convergence of either the mean and variance or confidence intervals. As the shape of the posterior is non-normally in this case, sample-based confidence interval approaches can describe the overall performance better than tracking the first two statistical moments. \\ Brooks and Gelman propose using diagnostics involving multiple simulations \cite{brooksgelman}, including a sample-based confidence interval approach. Here, the confidence interval is derived using two ways: First, the mean of the within-sequence confidence intervals is computed repeatedly with an increased number of samples. Second, the mean of the total-length confidence interval using multiple chains is tracked. Convergence is suggested when the ratio of the means approaches $1$ as it indicates that the confidence interval computed using the single Markov Chain matches the confidence interval that is based on the full information. For further exemplifications, please see \cite{brooksgelman}. \\
			The results of the method from Brooks and Gelman can be seen in Figure \ref{Brooks_Gelman_method2} for the cRW for model 2 and 3. It shows, that the single chain only has to contain around $1000$ samples for both models for the ratio approaching $1$. The results show that the cRW method still provides a converging Markov Chain for a model with multiple uncertain inputs.
			
			{\begin{figure}[H]
					\centering
					\includegraphics[width=0.8\linewidth]{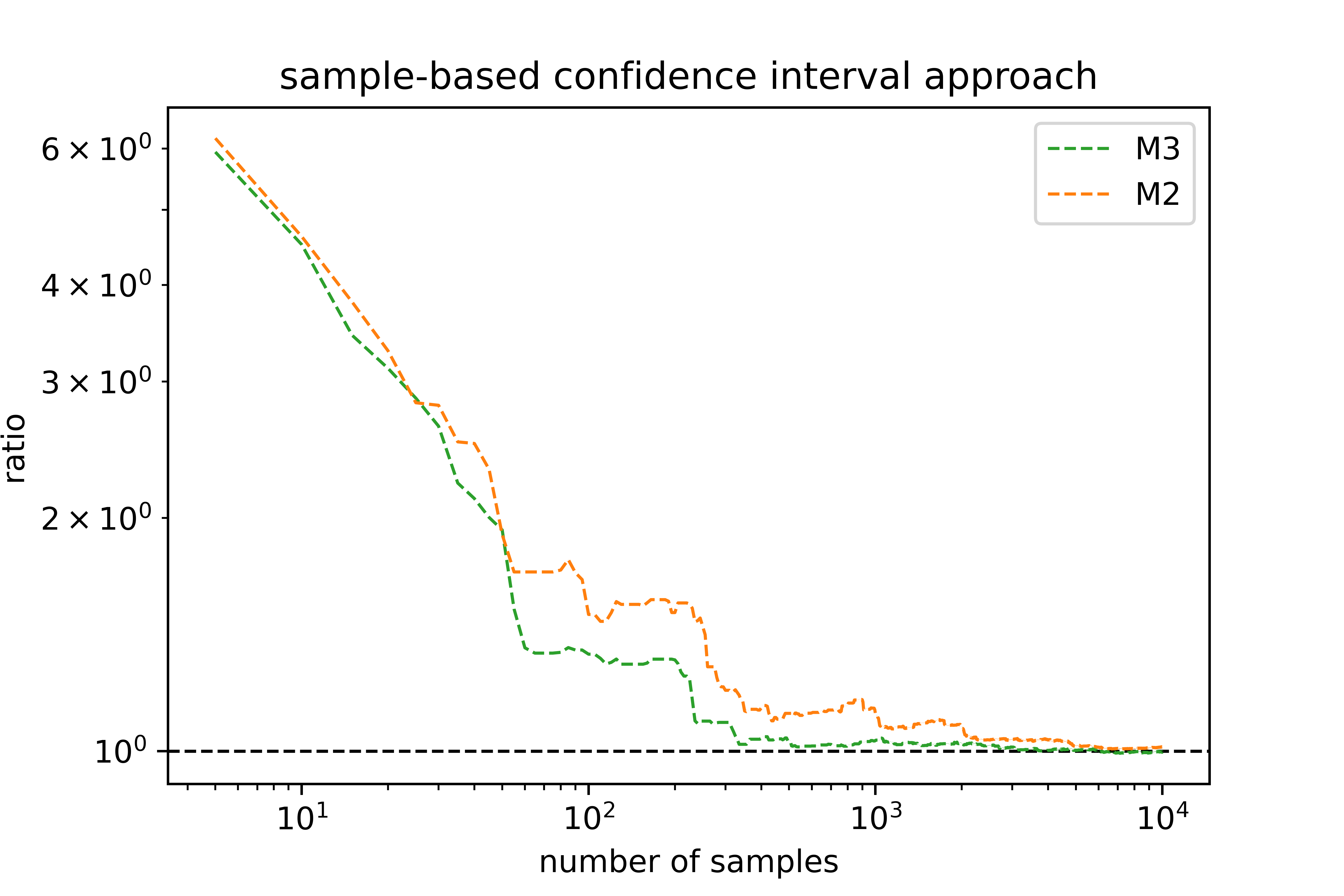}
					\caption{This figure shows the confidence interval based approach by Brooks and Gelman for the second (orange) and third (green) model. The ratio of the $95\%$ confidence interval of a single Markov Chain to the result using two independent Markov Chains of 500k samples is computed. }
					\label{Brooks_Gelman_method2}
			\end{figure} }
			
			\section{Conclusion}
			This work is motivated by transpiration cooling for rocket thrust chambers, where the cooling technique is used to reduce thermal loads of the combustion chamber walls. While being a promising cooling technique, the multiple factors influencing transpiration cooling challenge successful simulation, which motivates the integration of parametric uncertainties into the modeling. \\ 
			An approach is presented that extends Bayesian Inversion to the probabilistic constrained case using generalized Polynomial Chaos and Monte Carlo based methods to solve constrained inverse problems. The fusion enables a close monitoring of the system's critical response, the temperature, by a gPC expansion of the stochastic temperature system while at the same time leveraging the access to inverse problems through the Bayesian framework. \\
			This setting is put into use to constrain three sampling strategies, the RW, the Hamiltonian Monte Carlo and the Stein Variational Gradient Descent. The results of the one-dimensional and two-dimensional models show that the extended strategies successfully sample the constrained target measure and the framework can be extended to different models and constraint requirements. The proposed constrained RW shows a strong performance compared to the other methods in terms of the overall decrease of the $L2$ error over the size of the Markov Chain, as well as quality stability of the proposed candidates with respect to the constrained case and CPU time. \\
			To evaluate the chance constraint, a quantile of the temperature distribution needs to be estimated. Future work will deal with questions on how to improve the accuracy and efficiency of especially high quantile estimates. Extensions of the proposed probabilistic method to coupled simulations of the porous medium with a hot gas flow will be studied.
			
			\bibliographystyle{unsrt}
			\bibliography{/Users/ellasteins/Desktop/Literature/all_Literature}{}

\begin{thebibliography}{10}

\bibitem{arai2013}
M.~Arai and T.~Suidzu.
\newblock Porous ceramic coating for transpiration cooling of gas turbine
  blade.
\newblock {\em Journal of Thermal Spray Technology}, 22(5):690--698, 2013.

\bibitem{esser2016innovative}
B.~Esser, J.~Barcena, M.~Kuhn, A.~Okan, L.~Haynes, S.~Gianella, A.~Ortona,
  V.~Liedtke, D.~Francesconi, and H.~Tanno.
\newblock Innovative thermal management concepts and material solutions for
  future space vehicles.
\newblock {\em Journal of Spacecraft and Rockets}, 53(6):1051--1060, 2016.

\bibitem{Langener2011}
T.~Langener.
\newblock {\em A {C}ontribution to {T}ranspiration {C}ooling for {A}erospace
  {A}pplications using {CMC} {W}alls}.
\newblock PhD thesis, Universit\"at Stuttgart, 2011.

\bibitem{dlr129995}
F.~Strauss, J.~Witte, C.~Manfletti, and S.~Schlechtriem.
\newblock Experiments on {N}itrogen and {H}ydrogen {T}ranspiration {C}ooling in
  {S}upersonic {C}ombustion {R}amjets ({S}cramjets).
\newblock In {\em Space Propulsion 2018}, Mai 2018.
\newblock SP2018\_00113.

\bibitem{Dahmen2014}
W.~Dahmen, T.~Gotzen, S.~M{\"u}ller, and M.~Rom.
\newblock Numerical simulation of transpiration cooling through porous
  material.
\newblock {\em J. Numer. Meth. Fluids, 76(6):331-365}, 2014.

\bibitem{huang2017experimental}
G.~Huang, Y.~Zhu, Z.~Liao, X.-L. Ouyang, and P.-X. Jiang.
\newblock Experimental investigation of transpiration cooling with phase change
  for sintered porous plates.
\newblock {\em International Journal of Heat and Mass Transfer},
  114:1201--1213, 2017.

\bibitem{boehrk2014transpiration}
H.~B{\"o}hrk.
\newblock Transpiration cooling at hypersonic flight-{AKT}i{V} on {SHEFEX}
  {II}.
\newblock In {\em 11th AIAA/ASME Joint Thermophysics and Heat Transfer
  Conference}, page 2676, 2014.

\bibitem{Eckert1954ComparisonOE}
E.~Eckert and N.~Livingood.
\newblock Comparison of effectiveness of convection-, transpiration-, and
  film-cooling methods with air as coolant.
\newblock Technical report, National advisory committee for aeronautics, 1954.

\bibitem{selzer2009production}
M.~Selzer, T.~Langener, H.~Hald, and J.~von Wolfersdorf.
\newblock Production and characterization of porous {C/C} material.
\newblock {\em Sonderforschungsbereich Transregio 40--Annual Report 2009},
  2009.

\bibitem{wang2004experimental}
J.~Wang, J.~Messner, and H.~Stetter.
\newblock An experimental investigation on transpiration cooling part ii:
  comparison of cooling methods and media.
\newblock {\em International Journal of Rotating Machinery}, 10(5):355--363,
  2004.

\bibitem{Boehrk2010}
H.~Böhrk, O.~Piol, and M.~Kuhn.
\newblock Heat {B}alance of a {T}ranspiration-{C}ooled {H}eat {S}hield.
\newblock {\em Journal of Thermophysics and Heat Transfer}, 24(3):581--588,
  2010.

\bibitem{DING2019422}
R.~Ding, J.~Wang, F.~He, G.~Dong, and L.~Tang.
\newblock Numerical investigation on the performances of porous matrix with
  transpiration and film cooling.
\newblock {\em Applied Thermal Engineering}, 146:422--431, 2019.

\bibitem{leontiev2019effect}
A.~Leontiev, A.~Saveliev, B.~Kichatov, A.~Kiverin, A.~Korshunov, and
  V.~Sudakov.
\newblock Effect of gaseous coolant temperature on the transpiration cooling
  for porous wall in the supersonic flow.
\newblock {\em International Journal of Heat and Mass Transfer}, 142:118433,
  2019.

\bibitem{christopher2020dns}
N.~Christopher, J.~Peter, M.~Kloker, and J.-P. Hickey.
\newblock {DNS} of turbulent flat-plate flow with transpiration cooling.
\newblock {\em Int. J. Heat Mass Transfer}, 157:119972, 2020.

\bibitem{reimer2011transpiration}
T.~Reimer, M.~Kuhn, A.~G{\"u}lhan, B.~Esser, M.~Sippel, and A.~van Foreest.
\newblock Transpiration cooling tests of porous {CMC} in hypersonic flow.
\newblock In {\em 17th AIAA International Space Planes and Hypersonic Systems
  and Technologies Conference}, page 2251, 2011.

\bibitem{Selzer2014}
M.~Selzer, S.~Schweikert, and H.~Hald.
\newblock Throughflow characteristics of {C/C}.
\newblock {\em Sonderforschungsbereich Transregio 40--Annual Report}, 2014.

\bibitem{wu2018optimization}
N.~Wu, J.~Wang, F.~He, L.~Chen, and B.~Ai.
\newblock Optimization transpiration cooling of nose cone with non-uniform
  permeability.
\newblock {\em International Journal of Heat and Mass Transfer}, 127:882--891,
  2018.

\bibitem{huang2018transpiration}
G.~Huang, Z.~Min, Li. Yang, P.-X. Jiang, and M.~Chyu.
\newblock Transpiration cooling for additive manufactured porous plates with
  partition walls.
\newblock {\em International Journal of Heat and Mass Transfer},
  124:1076--1087, 2018.

\bibitem{min2019experimental}
Z.~Min, G.~Huang, S.~Parbat, L.~Yang, and M.~Chyu.
\newblock Experimental investigation on additively manufactured transpiration
  and film cooling structures.
\newblock {\em Journal of Turbomachinery}, 141(3):031009, 2019.

\bibitem{liu2013effects}
Y.-Q. Liu, Y.-B. Xiong, P.-X. Jiang, Y.-P. Wang, and J.-G. Sun.
\newblock Effects of local geometry and boundary condition variations on
  transpiration cooling.
\newblock {\em International Journal of Heat and Mass Transfer}, 62:362--372,
  2013.

\bibitem{wilcox2006turbulence}
D.~Wilcox.
\newblock Turbulence modeling for {CFD}. {L}a {C}anada, {CA}: {DCW}
  {I}ndustries.
\newblock {\em Inc, November}, 2006.

\bibitem{Koenig}
V.~K{\"o}nig, M.~Rom, and S.~M{\"u}ller.
\newblock {\em A {C}oupled {T}wo-{D}omain {A}pproach for {T}ranspiration
  {C}ooling. In {F}uture space-transport-system components under high thermal
  and mechanical loads : results from the {DFG} {C}ollaborative {R}esearch
  {C}enter {TRR40}}, chapter~2, pages 33 -- 49.
\newblock Springer, 2021.

\bibitem{cerminara}
A.~Cerminara, R.~Deiterding, and N.~Sandham.
\newblock Direct numerical simulation of hypersonic flow through regular and
  irregular porous surfaces.
\newblock {\em 7th European Conference on Computational Fluid Dynamics}, 2019.

\bibitem{mackay1998introduction}
D.~Mackay.
\newblock Introduction to monte carlo methods.
\newblock In {\em Learning in Graphical Models}, pages 175--204. Springer,
  1998.

\bibitem{xiu2005high}
D.~Xiu and J.~Hesthaven.
\newblock High-order collocation methods for differential equations with random
  inputs.
\newblock {\em SIAM Journal on Scientific Computing}, 27(3):1118--1139, 2005.

\bibitem{wiener1938homogeneous}
N.~Wiener.
\newblock The homogeneous chaos.
\newblock {\em American Journal of Mathematics}, 60(4):897--936, 1938.

\bibitem{cameron1947orthogonal}
R.~Cameron and W.~Martin.
\newblock The orthogonal development of non-linear functionals in series of
  {F}ourier-{H}ermite functionals.
\newblock {\em Annals of Mathematics}, pages 385--392, 1947.

\bibitem{ghanem1998probabilistic}
R.~Ghanem.
\newblock Probabilistic characterization of transport in heterogeneous media.
\newblock {\em Computer Methods in Applied Mechanics and Engineering},
  158(3-4):199--220, 1998.

\bibitem{xiu2002wiener}
D.~Xiu and G.~Karniadakis.
\newblock The {W}iener-{A}skey polynomial chaos for stochastic differential
  equations.
\newblock {\em SIAM Journal on Scientific Computing}, 24(2):619--644, 2002.

\bibitem{sudret2008global}
B.~Sudret.
\newblock Global sensitivity analysis using polynomial chaos expansions.
\newblock {\em Reliability Engineering \& System Safety}, 93(7):964--979, 2008.

\bibitem{Wu2019}
J.~Wu, J.~Wang, and S.~Shadden.
\newblock Adding {C}onstraints to {B}ayesian {I}nverse {P}roblems.
\newblock {\em The Thirty-Third AAAI Conference on Artificial Intelligence
  (AAAI-19)}, 2019.

\bibitem{Frazier2018}
P.~Frazier.
\newblock A {T}utorial on {B}ayesian {O}ptimization.
\newblock {\em arXiv:1807.02811}, 2018.

\bibitem{Gardner2014}
J.~Gardner, M.~Kusner, Z.~Xu, K.~Weinberger, and J.~Cunningham.
\newblock Bayesian {O}ptimization with {I}nequality {C}onstraints.
\newblock {\em Proceedings of the 31st International Conference on Machine
  Learning, Beijing, China}, 2014.

\bibitem{Gelbert2015}
M.~Gelbart.
\newblock {\em Constrained {B}ayesian {O}ptimization and {A}pplications}.
\newblock PhD thesis, Harvard University, 2015.

\bibitem{Peichl2021}
J.~Peichl, A.~Schwab, M.~Selzer, H.~Böhrk, and J.~von Wolfersdorf.
\newblock {\em Innovative Cooling for Rocket Combustion Chambers}, chapter~3,
  pages 51 -- 64.
\newblock Springer, 2021.

\bibitem{lu2015limitations}
F.~Lu, M.~Morzfeld, X.~Tu, and A.~Chorin.
\newblock Limitations of polynomial chaos expansions in the {B}ayesian solution
  of inverse problems.
\newblock {\em Journal of Computational Physics}, 282:138--147, 2015.

\bibitem{brooks2011handbook}
S.~Brooks, A.~Gelman, G.~Jones, and X.~Meng.
\newblock {\em Handbook of Markov Chain Monte Carlo}.
\newblock CRC press, 2011.

\bibitem{Kaipio2005}
J.~Kaipio and E.~Somersalo.
\newblock {\em {S}tatistical and {C}omputational {I}nverse {P}roblems}.
\newblock 0066-5452. Springer-Verlag New York, 2005.

\bibitem{wang2013adaptive}
Z.~Wang, S.~Mohamed, and N.~Freitas.
\newblock Adaptive {H}amiltonian and {R}iemann {M}anifold {M}onte {C}arlo.
\newblock In {\em International Conference on Machine Learning}, pages
  1462--1470. PMLR, 2013.

\bibitem{liu2016stein}
Q.~Liu and D.~Wang.
\newblock Stein variational gradient descent: {A} general purpose bayesian
  inference algorithm.
\newblock {\em arXiv preprint arXiv:1608.04471}, 2016.

\bibitem{Somersalo2007}
D.~Calvetti and E.~Somersalo.
\newblock {\em Introduction to {B}ayesian {S}cientific {C}omputing}.
\newblock Springer, 2007.

\bibitem{Burke}
J.~Burke.
\newblock The gradient projection algorithm. 2014.
\newblock Lecture notes at the University of Washington.

\bibitem{Dahmen2015}
M.~Rom S. Schweikert M. Selzer J. von~Wolfersdorf W.~Dahmen, S.~M{\"u}lller.
\newblock Numerical boundary layer investigations of transpiration-cooled
  turbulent channel flow.
\newblock {\em International Journal of Heat and Mass Transfer}, 86:90--100,
  2015.

\bibitem{brooksgelman}
S.~Brooks and A.~Gelman.
\newblock General methods for monitoring convergence of iterative simulations.
\newblock {\em Journal of Computational and Graphical Statistics},
  7(4):434--455, 1998.

\end{thebibliography}
			
			\section{Acknowledgement}
			This work was supported by the Deutsche Forschungsgemeinschaft (DFG, German Research Foundation) - 333849990/GRK2379 (IRTG Modern Inverse Problems).
			
			\appendix
			\section{Appendix}
			
			\begin{table}[H]
				\caption{Deterministic parameter set (with dimensions)}	
				\begin{tabular}[h]{l|l|l}	
					Parameter & Symbol & Value \\	
					\hline
					Reynolds number & $Re$ & $Re_0=405$ \\
					Prandtl number & $Pr_f$ & $0.64$ \\
					Nusselt number & $Nu_{v,f}$ & $7500$ \\
					Heat flux & $q$  & $q_0=30845 \, \, [\frac{W}{m^2}]$\\
					Hot gas temperature & $T_{HG}$ & $347 \, \, [K]$ \\
					Porosity & $\phi$ & $\phi_0=0.111, \phi_1=0.4$ \\
					Thermal conductivity of fluid & $\kappa_f$ & $0.03 \, \, [\frac{W}{mK}]$ \\
					Thermal conductivity of solid & $\kappa_s$ & $15.2 \, \, [\frac{W}{mK}]$ \\
					Permeability & $K_D$ & $3.57e-13 \, \, [m^2]$ \\ 
					Forchheimer coefficient & $K_F$ & $5.17e-08 [m]$ \\
					Coolant reservoir temperature & $T_c$ & $304.2 \, \, [K]$ \\
					Solid reservoir temperature & $T_b$ & $321.9 \, \, [K]$ \\
					Reservoir pressure & $p_R$ & $600000 \, \, [Pa]$ \\
					Length of the porous medium & $L$ & $0.015 \, \, [m]$ 
					\label{tab_set}
				\end{tabular}
			\end{table}
			
			\begin{algorithm}[H]
				\caption {constrained random walk Markov Chain Monte Carlo (cRW)} \label{cMCMC_alg} 
				\begin{algorithmic}
					\State Given: $\sigma$ and initial value $\theta_0$ 
					\For{$i=0,...,N-1$}
					\State 	Draw candidate $\theta^*|\theta_{i}, \text{where } \theta^* \sim \mathcal{N}(\theta_{i}, \sigma^2)$ 
					\State  $\chi(\theta^*)=  \left\{ \begin{matrix} 1 &\text{ if } P_{\boldsymbol{\xi}}(f_2^{gPC}(\boldsymbol{\xi}; \theta^*) \leq \beta) \geq \alpha \\ 0 &\text{ otherwise} \end{matrix} \right.$ 
					\State 	$\alpha_1=\chi_\mathcal{S}(\theta^*) \cdot \text{min } (1, \frac{\Pi_{l}(\theta^*) \cdot \Pi_{prior}(\theta^*)}{\Pi_{l}(\theta_{i}) \cdot \Pi_{prior}(\theta_{i}) })$ 
					\EndFor
					\If {$\alpha_1 = 1$:}
					\State Accept candidate and set $\theta_{i+1}=\theta^*$
					\Else 
					\State Accept candidate and set $\theta_{i+1}=\theta^*$ with probability $\alpha_1$, \\ 
					           else reject and set $\theta_{i+1}=\theta_{i}$ 
					\EndIf
				\end{algorithmic}
			\end{algorithm} 
			
			\begin{algorithm}[H]
				\caption {constrained Hamiltonian Monte Carlo (cHMC)} \label{cHMC_alg} 
				\begin{algorithmic}
					\State Given: $\mathcal{G}$, $\Pi_{prior}$, $\Pi_{l}$, $L$, $s$, $M$ and initial value $\theta_0$ 
					\For{ $i=0,...,N-1$} 
					\State Draw $\tilde{m} \sim \mathcal{N}(0,M)$ and $L_1 \sim \mathcal{U}(1,L)$
					\State Let $\theta_0=\theta_i$ and $\tilde{m} _0= \tilde{m}  + \frac{s}{2} \nabla \mathcal{G}(\theta)|_{\theta_0}$
					\For{$ j= 1,...,L_1$} \\  \textit{Leapfrog-Integration}
					\State $\theta_j = \theta_{j-1} + s \cdot M^{-1} \cdot \tilde{m} _j$ 
					\State $\tilde{m} _j = \tilde{m} _{j-1} + s \cdot \nabla \mathcal{G}(\theta)|_{\theta_j}$ 
					\EndFor
					\State $\theta^* = \theta_{L_1}$
					\State $\tilde{m} ^* = \tilde{m} _{L_1-1} - \frac{s}{2} \nabla \mathcal{G}(\theta)|_{\theta^*}$
					\State \textit{compute acceptance ratio:} \\ 
					$\alpha_1=\text{min } (1, \exp[{\log{(\Pi_{prior}(\theta^*) \cdot \Pi_{l}(\theta)^*)} -\frac{1}{2} M^{-1} (\tilde{m} ^*)^2 - \log{(\Pi_{prior}(\theta) \cdot \Pi_{l}(\theta))} + \frac{1}{2} M^{-1} (\tilde{m} )^2}])$ 
					\If{$\alpha_1 = 1$}
					\State Accept candidate and set $\theta_{i+1}=\theta^*$ 
					\Else
					\State Set $\theta_{i+1}=\theta^*$ with probability $\alpha_1$, else reject and set $\theta_{i+1}=\theta_{i}$. 	
					\EndIf
					\EndFor
				\end{algorithmic}
			\end{algorithm} 
			
			\begin{algorithm}[H]
				\caption {constrained Stein Variational Gradient Descent (cSVGD) } \label{cSVGD_alg} 
				\begin{algorithmic}
					\State Given: $\mathcal{L}$, a RBF kernel $k$, initial step size $s_1$, and a set of initial particles $\{\theta_i^0\}_{i=1}^{N_1}$
					\For {$l=0,...,N_2-1$}
					\State $\theta_i^{l+1}=\theta_i^l+ s_l \, \left( \frac{1}{n} \sum_{j=1}^n [k(\theta_j^l, \theta) \, \nabla_{\theta_j^l} \mathcal{L}(\theta_j^l) + \nabla_{\theta_j^l} \, k(\theta_j^l, \theta)] \right)$
					\EndFor
				\end{algorithmic}
			\end{algorithm} 
			
		\end{document}